\documentclass[12pt]{article}
\usepackage[usenames,dvipsnames,svgnames,table]{xcolor} 
\usepackage{kotex}
\usepackage{graphicx}
\usepackage{epsfig}
\usepackage{amssymb, amsmath, mathtools,amsthm}
\usepackage{verbatim}
\usepackage{natbib}
\usepackage[affil-it]{authblk}
\usepackage{mathrsfs}
\usepackage{here}
\usepackage{makecell}
\usepackage{tikz}
\usepackage[shortlabels]{enumitem}
\usepackage{yfonts}
\usepackage{comment}
\usepackage{longtable}
\usepackage{tikz,amsmath, amssymb,bm,color}
\usetikzlibrary{circuits.logic.US,circuits.logic.IEC,fit}
\usepackage{multirow}
\usepackage{adjustbox}
\usepackage{algorithmicx}
\usepackage[ruled]{algorithm}
\usepackage{algpseudocode}

\usepackage{caption}
\usepackage{arydshln,leftidx}
\usepackage{booktabs}
\usepackage{rotating}
\usepackage{lipsum}
\usepackage{appendix}
\usepackage{titlesec}
 \usepackage{etoolbox, chngcntr}
\AtBeginEnvironment{appendices}{%
 \titleformat{\section}{\bfseries\Large}{\appendixname~\thesection.}{0.5em}{}%
 \titleformat{\subsection}{\bfseries}{\thesubsection}{0.5em}{}%
\counterwithin{equation}{section}
}

\usepackage[colorlinks=true, citecolor=black, urlcolor=black]{hyperref}
\usepackage[colorinlistoftodos, textsize=scriptsize]{todonotes} 

\newtheorem{theorem}{Theorem}[section]
\newtheorem{lemma}{Lemma}[section]
\newtheorem{cor}{Corollary}[section]
\newtheorem{prop}{Proposition}[section]
\newtheorem{remark}{Remark}[section]
\DeclareMathOperator*{\argmax}{argmax}
\newcommand{\bsDelta}{\boldsymbol{\Delta}}

\newcommand*\interior[1]{\mathring{#1}}

\newcommand{\bfs}{ \mathbf{s}}
\newcommand{\bfu}{ \mathbf{u}}
\newcommand{\calA}{\mathcal{A}}

\newcommand{\calC}{\mathcal{C}}
\newcommand{\calD}{\mathcal{D}}

\newcommand{\calF}{\mathcal{F}}
\newcommand{\calG}{\mathcal{G}}

\newcommand{\calK}{\mathcal{K}}

\newcommand{\calM}{\mathcal{M}}
\newcommand{\calP}{\mathcal{P}}
\newcommand{\calQ}{\mathcal{Q}}

\newcommand{\calS}{\mathcal{S}}

\newcommand{\calU}{\mathcal{U}}

\newcommand{\calZ}{\mathcal{Z}}
\newcommand{\kommentar}[1]{}
\def\Real{\hbox{I\kern-.1667em\hbox{R}}}
\def\Reals{\hbox{\scriptsize I\kern-.1667em\hbox{R}}}

\newcommand{\bsu}{\boldsymbol{u}}

\newcommand{\bss}{\boldsymbol{s}}

\newcommand{\bsbeta}{\boldsymbol{\beta}}

\newcommand{\bsgamma}{\boldsymbol{\gamma}}
\newcommand{\bssig}{\boldsymbol{\sigma}}

\newcommand{\bsLambda}{\boldsymbol{\Lambda}}
\newcommand{\bsTheta}{\boldsymbol{\Theta}}
\newcommand{\bsPhi}{\boldsymbol{\Phi}}
\newcommand{\diag}{\mbox{diag}}
\newcommand{\bsSig}{\boldsymbol{\Sigma}}

\newcommand{\bsOmega}{\boldsymbol{\Omega}}
\newcommand{\bszeta}{\boldsymbol{\zeta}}
\newcommand{\bsPsi}{\boldsymbol{\Psi}}
\newcommand{\bsxi}{\boldsymbol{\xi}}

\newcommand{\bfX}{{\bf X}}
\newcommand{\bfA}{{\bf A}}
\newcommand{\bfB}{{\bf B}}
\newcommand{\bfx}{{\bf x}}
\newcommand{\bfV}{{\bf V}}

\newcommand{\bfS}{{\bf S}}
\newcommand{\bfU}{{\bf U}}
\newcommand{\bfE}{{\bf E}}
\usepackage{dsfont}
\DeclareMathOperator{\trace}{tr}
\DeclareMathOperator{\frob}{F}
\newcommand{\I}{\textbf{I}}
\DeclareMathOperator{\edge}{E}
\DeclareMathOperator{\hessian}{H}
\DeclareMathOperator{\vectorize}{vec}
\DeclareMathOperator*{\minimize}{Minimize}
\DeclareMathOperator{\bigO}{O}

\newcommand{\diff}{\text{d}}
\newcommand{\bea}{\begin{eqnarray*}}
   \newcommand{\eea}{\end{eqnarray*}}
\newcommand{\bean}{\begin{eqnarray}}
\newcommand{\eean}{\end{eqnarray}}
\newcommand{\benu}{\begin{enumerate}}
\newcommand{\eenu}{\end{enumerate}}

\newcommand{\bbR}{\mathbb{R}}

\newcommand{\bbP}{\mathbb{P}}

\newcommand{\cM}{\mathcal{M}}

\newcommand{\Exp}{\text{Exp}}
\newcommand{\N}{\text{N}}
\newcommand{\eigmin}{\lambda_{\min}}
\newcommand{\eigmax}{\lambda_{\max}}
\newcommand{\sgij}{\sigma_{ij}}
\newcommand{\sgii}{\sigma_{ii}}
\newcommand{\deltaij}{\Delta_{\calZ,ij}}
\newcommand{\deltaii}{\Delta_{\calZ,ii}}
\newcommand{\omegaij}{\omega_{\calZ,ij}}

\newcommand{\psiij}{\psi_{ij}}

\newcommand{\bfzero}{\mathbf{0}}

\newcommand{\dsone}{\mathds{1}}
\newcommand{\zij}{z_{ij}}

\makeatletter
\newcommand*{\rom}[1]{\expandafter\@slowromancap\romannumeral #1@}
\makeatother

\title{Covariance Structure Estimation with Laplace Approximation}
\author[1]{Bongjung Sung}
\author[2]{Jaeyong Lee}
\affil[1,2]{Department of Statistics\\ Seoul National University}

\begin{document}\setlength {\marginparwidth }{2cm}

\maketitle 
 
\begin{abstract}
 Gaussian covariance graph model is a popular model in revealing underlying dependency structures among random variables. A Bayesian approach to the estimation of covariance structures uses priors that force zeros on some off-diagonal entries of covariance matrices and put a positive definite constraint on matrices. In this paper, we consider a spike and slab prior on off-diagonal entries, which uses a mixture of point-mass and normal distribution. The point-mass naturally introduces sparsity to covariance  structures so that the resulting posterior from this prior renders covariance structure learning. Under this prior, we calculate posterior model probabilities of covariance structures using Laplace approximation. We show that the error due to Laplace approximation becomes asymptotically marginal at some rate depending on the posterior convergence rate of covariance matrix under the Frobenius norm. With the approximated posterior model probabilities, we propose a new framework for estimating a covariance structure. Since the Laplace approximation is done around the mode of conditional posterior of covariance matrix, which cannot be obtained in the closed form, we propose a block coordinate descent algorithm to find the mode and show that the covariance matrix can be estimated using this algorithm once the structure is chosen. Through a simulation study based on five numerical models, we show that the proposed method outperforms graphical lasso and sample covariance matrix in terms of root mean squared error, max norm, spectral norm, specificity, and sensitivity. Also, the advantage of the proposed method is demonstrated in terms of accuracy compared to our competitors when it is applied to linear discriminant analysis (LDA) classification to breast cancer diagnostic dataset.
\end{abstract}

\section{Introduction}\label{sec:intro}
The sparse covariance matrix estimation problem based on a high-dimensional dataset, where the sample size $n$ exceeds the dimension of variable $p$, has been grown in importance in multivariate data analysis since the problem is crucial to uncover underlying dependency structures among $p$-dimensional variables. The sparse covariance matrix estimation plays a key role in many multivariate statistical inferences, such as principal component analysis (PCA), linear discriminant analysis (LDA), and time series analysis. In the estimation, Gaussian covariance graph model provides an excellent tool, assuming multivariate normal distribution on data. Under the normality, the covariance matrix induces a bi-directed graph and the absence of an edge between two variables is equivalent to zero on the covariance between them. Hence, inference of covariance structure, or equivalently graphical structure, can lead to the estimation of covariance matrix. But the inference requires introducing sparsity to the structure, as the structure is determined by the zeros in the covariance matrix.

In the frequentist literature, introducing sparsity to the structure has been mostly done using regularization methods. \cite{bien2011sparse} and \cite{yuan2007} considered $\ell_1$-type penalty on the negative log-likelihood. They introduced sparsity to the structure of the model using sparse and shrinkage estimators. On the other hand, \cite{bickel2008covariance}, \cite{rothman2009generalized}, and \cite{cai2011adaptive} considered thresholding method. However, the estimated covariance matrices from these methods sometimes contradicted with a positive definite constraint on covariance matrices. Other works in the frequentist context include banding and tapering, e.g., \cite{wu2003}, \cite{bickel2008regularization}, and \cite{kauf2012}. Their estimators are often used for Gaussian covariance graph model and are well supported by the asymptotic statistical properties.

The Bayesian methods for introducing sparsity to the covariance matrix have been also developed. The G-inverse Wishart prior, considered by \cite{silva2009hidden}, was often used in the Bayesian framework. However, the method had a limitation in applications as the dimension gets higher because of posterior intractability. \cite{khare2011wishart} extended G-inverse Wishart prior to a broader class and proposed a blocked Gibbs sampler to sample covariance matrices from the resulting posterior, but \cite{khare2011wishart} considered only decomposable graphs. \cite{wang15} also considered covariance graph model for the inference of covariance matrices and also provided a blocked Gibbs sampler that is applicable to all graphs but under the continuous spike and slab prior, which uses the mixture of two Gaussian distributions on off-diagonal entries, one with sufficiently small variance and the other with variance substantially far from $0$, and the exponential prior on diagonal entries. But the method proposed by \cite{wang15} had inherent difficulty in introducing sparsity to the structure of the model due to the absolute continuity of prior. Furthermore, none of \cite{silva2009hidden}, \cite{khare2011wishart}, and \cite{wang15} did not attain any asymptotic statistical properties. \cite{lee2021} were able to introduce sparsity to the covariance structure using beta-mixture shrinkage prior and attained the posterior convergence rate of covariance matrix under the Frobenius norm. However, compared to frequentist literature, Bayesian literature still lacks in the methods for inference of sparse covariance matrices, and most of the proposed methods are not theoretically well supported. Furthermore, they have a limit in introducing sparsity to the covariance structures due to the absolute continuity of prior.

To fill the gap in the literature of Bayesian inference, we propose a method for estimating covariance structures under spike and slab prior which uses a mixture of point-mass and normal distribution and exponential prior on off-diagonal entries and diagonal entries of covariance matrices, respectively, which is a modified version of the prior in \cite{wang15}. To overcome the limitations of the current Bayesian inference on sparse covariance matrices we have described, we propose a method that uses Laplace approximation to compute posterior probabilities of covariance structures, generates MCMC samples of graphs using Metropolis-Hastings algorithm proposed by \cite{liu2019empirical} and chooses the model either by median probability model (MPM) or maximum a posteriori (MAP). We estimate covariance matrix by the model of conditional posterior of covariance matrix given the structure.

Because of the enforced zero due to point-mass prior, a blocked Gibbs sampler for sampling covariance matrix from the resulting posterior or reversible jump MCMC (RJMCMC) for computing posterior model probabilities of covariance structures is not suitable in this case. To be specific, the enforced zero due to point-mass makes it difficult to derive conditional posteriors induced from the proposed prior if one considers the blocked Gibbs sampler proposed by \cite{wang15}, hence bthe locked Gibbs sampler is not applicable to this case. Furthermore, if point-mass is introduced, there are $2^{\binom{p}{2}}$ models that RJMCMC has to visit over, which makes practical implementation extremely difficult as $p$ grows and the estimated posterior model probability of the covariance structure unreliable. We show that the error by Laplace approximation becomes asymptotically marginal at some rate depending on the posterior convergence rate of covariance matrix under the Frobenius norm. 

One of the advantages of using the suggested prior in this paper is that the posterior is always twice continuously differentiable. In the estimation of sparse precision matrices, \cite{banerjee2015bayesian} proposed the Bayesian version of graphical lasso, which uses Laplace approximation to compute posterior model probabilities of graphical structures and chooses the final model by MPM. \cite{banerjee2015bayesian} considered Laplace prior on off-diagonal entries. Since Laplace prior is not differentiable at the median, Laplace approximation was applicable to only regular models. The term regularity may follow the terminology used by \cite{yuan2005} and \cite{banerjee2015bayesian}. \cite{yuan2005} and \cite{banerjee2015bayesian} put $\ell_1$-type penalty on the entries of the concentration matrix, which is the graphical lasso. When an off-diagonal entry is set as zero due to graphical lasso, the integrand in Laplace approximation becomes non-differentiable. If the model includes such an off-diagonal entry as a free variable, \cite{yuan2005} and \cite{banerjee2015bayesian} referred to such model as a non-regular model, and a regular model otherwise. \cite{banerjee2015bayesian} showed that the posterior probability of regular models is not smaller than that of their non-regular counterparts so that one can consider only regular models if one is to choose a model by MPM. However, one had to determine the regularity of each model and it turned out that the ratio of regular models among all possible models is extremely low. Moreover, the constraint regularity led the estimated matrix to be extremely sparse. On the contrary, since the posterior induced from the prior in this paper is always twice continuously differentiable, we do not have to consider the regularity. Thus, we do not have to suffer from the drawback of \cite{banerjee2015bayesian} we just mentioned. Also, even though \cite{banerjee2015bayesian} put a limit on the number of edges of a graph by Frequentist graphical lasso, there were still too many models to search over and one had to judge the regularity of each model which yields inefficiency in choosing the model. However, by using the algorithm in \cite{liu2019empirical}, we do not have to search over all models. Furthermore, since we force zeros on off-diagonal entries, sparsity can be naturally introduced to the structure of covariance matrix compared to those with continuous spike and slab prior in \cite{wang15} and \cite{lee2021}. Hence, our prior is more useful in introducing sparsity to the covariance structure.

The paper is organized as follows. In section \ref{sec:prior}, we introduce notations and preliminaries necessary for this paper and the prior considered in this paper. Then, in section \ref{sec:post}, we describe Laplace approximation to calculate posterior model probabilities of covariance structures and Metropolis-Hastings algorithm to generate MCMC samples of graphs from the resulting approximated posterior. We choose the final model from the MCMC samples either by MPM and MAP. Also, we show that the error by Laplace approximation becomes asymptotically negligible at some rate depending on the posterior convergence rate of covariance matrix and propose a block descent algorithm to find the mode of conditional posterior of covariance matrix as the Laplace approximation is done around it, which cannot be obtained in the closed form. We describe how to estimate the covariance matrix when the covariance structure is given using this algorithm. Finally, in section \ref{sec:simul}, we provide the simulation results for five numerical models and breast cancer diagnostic dataset. The conclusion will be discussed in \ref{sec:discuss}. 

\section{Prior and Posterior Convergence Rate}\label{sec:prior}
\subsection{Notations and Preliminaries}

Consider the graph $G=V\times V$ where $V$ is the set of $p$ nodes and $E\subseteq V \times V$ is the set of edges. Let $\mathcal{Z}=(\zij)$ be a $p(p-1)/2$ edge inclusion vector, or equivalently a covairance (graphical) structure indicator, i.e., $\zij=1$ if $(i,j)$ or $(j,i)\in E$ and $\zij=0$ otherwise, for $i<j$. Note that in this paper we consider bi-directed graphs, especially Gaussian covariance graph model, and thus $z_{ij}=1$ if and only if $z_{ij}$ for all $i,j\in V$ with $i<j$. Denote the sum of entries in $\mathcal{Z}$ by $\#\calZ$, the number of edges in $G$. Suppose two positive numerical sequences $a_n$ and $b_n$ are given. If $a_n/b_n$ is bounded as $n\rightarrow \infty$, we denote this by $a_n\lesssim b_n$ $(b_n\gtrsim a_n)$, or equivalently $b_n=\bigO(a_n)$. If $a_n\lesssim b_n$ and $b_n\lesssim a_n$ both hold, we write $a_n\asymp b_n$. 

Let $\cM$ be the set of all $p\times p$ real symmetric matrices. We denote the set of all $p\times p$ positive definite matrices by $\cM^+\subset \cM$. Suppose $\bfA=\left(a_{ij}\right)$ and $\bfB=\left(b_{ij}\right)\in \cM$. Denote the Hadamard product of $\bfA$ and $\bfB$ by $\bfA\circ\bfB$ and the Kronecker product by $\bfA\otimes\bfB$. Let $\eigmin\left(\bfA\right)$ and $\eigmax\left(\bfA\right)$ denote the smallest and the largest eigenvalues of $\bfA$, respectively. Also, let $||\bfA||_{\infty}=\max_{i,j}|a_{ij}|$, $||\bfA||_{\frob}=\sqrt{\sum_{i,j=1}^pa_{ij}^2}$, and denote the spectral norm by $||\cdot||_2$. Note the following facts hold: 
\begin{align}
    &||\bfA||_{\infty}\leq ||\bfA||_2\leq ||\bfA||_{\frob}\leq p||\bfA||_{\infty},\\
    &||\bfA\bfB||_{\frob}\leq ||\bfA||_2||\bfB||_{\frob}.
\end{align}
We write $\bfB\prec\bfA$ ($\bfA\succ\bfB$) if $\bfA-\bfB\in \calM^+$.

\subsection{Prior Setting}
Suppose we observe $n$ independent random samples $X_1,\dots,X_n|\bsSig$ from $\N_p(\bfzero,\bsSig)$, where $\bsSig=(\sgij)\in \cM^+$ and $\bsSig$ is sparse, i.e., many of $\sigma_{ij}$ are zeros. For Bayesian inference on the sparsity of $\bsSig$ in this paper, we consider a spike and slab prior on off-diagonal entries, which uses a mixture of point-mass and Gaussian distribution, and exponential prior on diagonal entries as follows:
\begin{align}
    \pi^u\left(\sgij\right)&=(1-q)\delta_0+q\N(\sgij|0,v^2), \quad 1\leq i<j \leq p,\\
    \pi^u\left(\sgii\right)&=\Exp\left(\sgii|\lambda/2\right), \quad 1\leq i \leq p,
\end{align}
\noindent where $v$ is some positive constant substantially far from $0$, $\lambda>0$, $q\in(0,1)$, and $\delta_0$ is point-mass. Here we consider $\lambda/2$ to be the rate parameter of the exponential distribution. By (3) and (4), the prior on the entries of $\Sigma$ is defined as 
\begin{align}
    \pi^u\left(\Sigma\right)=\prod_{i<j}((1-q)\delta_0+q\N(\sgij|0,v^2))\prod_{i=1}^p\Exp\left(\sgii|\lambda/2\right).
\end{align}

The prior (5) can be equivalently defined through a hierarchical model with an edge inclusion vector $\mathcal{Z}=(\zij)$,
\begin{align*}
    \pi^u(\Sigma|\calZ)&=\prod_{z_{ij}=1}\N(\sgij|0,v^2)\prod_{i=1}^p \Exp(\sgii|\lambda/2), \\
    \pi^u(\calZ)&=\prod_{i<j}\pi^{\zij}(1-\pi)^{1-\zij}.
\end{align*}
Thus, one can interpret prior (3) as $\sgij=0$ if $z_{ij}=0$ with probability $1-q$ and $\sgij\sim\N(\cdot|0,v^2)$ if $z_{ij}=1$ with probability $q$ for $i<j$. Hence, $\calZ$ can be seen as a covariance structure indicator, or equivalently a graphical structure indicator, following the terminology used by \cite{banerjee2015bayesian}. Note $q$ is the edge acceptance probability; and the larger $q$ is, the less model becomes sparse. We may choose not too small $q$ and $v$ to avoid the extreme sparsity of the model. Compared to the prior in \cite{wang15}, our prior naturally introduces sparsity to the structure of model with the point-mass at zero. Define the set
\begin{align}
    \mathcal{U}(\tau)=\{\calC\in \cM^+: 1/\tau \leq \lambda_{\min}\left(\calC\right)\leq \lambda_{\max} \left(\calC\right) \leq \tau \},
\end{align}
where $\tau>1$. We restrict $\Sigma$ on (6) and so consequently we have the following prior
\begin{align}
\pi(\Sigma)\propto \pi^u\left(\bsSig\right)\dsone(\Sigma \in \calU(\tau)).
\end{align}
Note that this restriction was frequently used in many statistical inferences on sparse covariance matrices. The restriction $\calU(\tau)$ is useful in deriving asymptotic statistical properties for our proposed method under regular conditions, though we consider $\tau$ as $\infty$ in a practical implementation.  

Compared to prior in \cite{wang12} or \cite{banerjee2015bayesian}, which considered Laplace prior that is not differentiable at its median on off-diagonal entries, prior in (7) is always twice continuously differentiable. Thus, if we are to compute the posterior model probabilities of $\calZ$ using Laplace approximation, prior on off-diagonal entries in (7) takes the advantage over Laplace prior because it is twice continuously differentiable and we do not need to check the regularity of the model. 

\section{Posterior Computation}\label{sec:post}
\subsection{Posterior}
Suppose $\bfX_n=(X_1,\dots,X_n)^t$ follows $\N_p(\bfzero,\bsSig)$ and $\bsSig$ follows prior (7). The likelihood of $\bfX_n$ is given $\bsSig$ as follows. 
\begin{align}
\begin{split}
    \pi(\bfX_n|\bsSig)&=\prod_{i=1}^n\frac{1}{|2\pi\bsSig|^{1/2}}\exp\left(-\frac{1}{2}X_i^t\bsSig^{-1}X_i\right)\\
    &\propto \prod_{i=1}^n|\bsSig|^{-1/2}\exp\left(-\frac{1}{2}\trace\left(X_i^t\bsSig^{-1}X_i\right)\right)\\
    &=\prod_{i=1}^n \exp\left(-\frac{1}{2}\log|\bsSig|-\frac{1}{2}\trace\left(X_iX_i^t\bsSig^{-1}\right)\right)\\
    &=\exp\left(-\frac{n}{2}\log|\bsSig|-\frac{1}{2}\trace\left(\sum_{i=1}^nX_iX_i^t\bsSig^{-1}\right) \right)\\
    &=\exp\left(-\frac{n}{2}\log|\bsSig|-\frac{n}{2}\trace\left(\bfS\bsSig^{-1}\right)\right),
\end{split}
\end{align} 
 where $\bfS=\frac{1}{n}\bfX_n^t\bfX_n=(s_{ij})$, which is the sample covariance matrix. 
Note 
\begin{align}
\pi\left(\bsSig|\calZ\right)\pi^u\left(\calZ\right)&\propto (1-q)^{p(p-1)/2-\#\calZ}\prod_{\zij=1}q\N(\sgij|0,v^2)\prod_{i=1}^p\Exp(\sgii|\frac{\lambda}{2})\dsone\left(\bsSig\in\calU_{\calZ}(\tau)\right)\nonumber\\
&\propto\left(\frac{q}{1-q}\frac{1}{\sqrt{2\pi}v}\right)^{\# \calZ}\exp\left(-\frac{1}{2v^2}\sum_{\zij=1}\sgij^2-\frac{\lambda}{2}\sum_{i=1}^p \sgii\right)\dsone\left(\bsSig\in\calU_{\calZ}(\tau)\right)\nonumber\\
&=\left(\frac{q}{1-q}\frac{1}{\sqrt{2\pi}v}\right)^{\# \calZ}\exp\left(-\frac{n}{2}p(\bsSig,\calZ)\right)\dsone\left(\bsSig\in\calU_{\calZ}(\tau)\right),
\end{align}
where $p(\bsSig,\calZ)=1/nv^2\sum_{\zij=1}\sgij^2+\lambda/n\sum_{i=1}^p \sgii$ and
\begin{align*}
    \calU_{\calZ}(\tau)=\{\bsSig=(\sgij)\in\calU(\tau):\sgij=0 \text{ if } \zij=0\}.
\end{align*}

Now, we obtain the marginal posterior of $\calZ$ under prior (7). By Bayes' rule, together with (8) and (9), we have the following conditional joint probability density function of $\calZ,\bsSig|\bfX_n$. 
\begin{align}
\pi\left(\calZ,\bsSig|\bfX_n\right)&\propto \pi\left(\bfX_n|\bsSig,\calZ\right)\pi\left(\bsSig|\calZ\right)\pi^u\left(\calZ\right) \nonumber \\
&\propto \left(\frac{q}{1-q}\frac{1}{\sqrt{2\pi}v}\right)^{\# \calZ}\exp\left(-\frac{n}{2}\log|\bsSig|-\frac{n}{2}\trace\left(\bfS\bsSig^{-1}\right)-\frac{n}{2}p(\bsSig,\calZ)\right)\dsone\left(\bsSig\in\calU_{\calZ}(\tau)\right)\nonumber \\
&= \left(\frac{q}{1-q}\frac{1}{\sqrt{2\pi}v}\right)^{\# \calZ}\exp\left(-\frac{n}{2}r_{\calZ}(\bsSig,\bfX_n)\right)\dsone\left(\bsSig\in\calU_{\calZ}(\tau)\right),
\end{align}
where $r_{\calZ}(\bsSig,\bfX_n)=\log|\bsSig|+\trace\left(\bfS\bsSig^{-1}\right)+p(\bsSig,\calZ)$. Let $\bar{E}_{\calZ}=\{(i,j):1\leq i=j\leq p \text{ or } z_{ij}=1\}$. For the notational simplicity, denote $\prod_{(i,j)\in\bar{E}_{\calZ}}d\sigma_{ij}$ by $d\bsSig_{\calZ}$. By (10), we have
\begin{align}
\begin{split}
    \pi(\calZ|\bfX_n)&\propto \int_{\bsSig_{\calZ}\in\calU_{\calZ}(\tau)}\pi\left(\calZ,\bsSig_{\calZ}|\bfX_n\right)d\bsSig_{\calZ}\\
    &\propto  \left(\frac{q}{1-q}\frac{1}{\sqrt{2\pi}v}\right)^{\# \calZ}\int_{\bsSig_{\calZ}\in\calU_{\calZ}(\tau)}\exp\left(-\frac{n}{2}r_{\calZ}\left(\bsSig_{\calZ},\bfX_n\right)\right)d\bsSig_{\calZ}.
\end{split}
\end{align}
Note that the posterior of $\calZ$ is very intractable. But as $r_{\calZ}(\bsSig,\bfX_n)$ is twice continuously differentiable and has a minimizer on the domain $\calU_{\calZ}(\tau)$, if we choose $\tau$ so that $\calU_{\calZ}(\tau)$ is broad enough to contain the minimizer, we approximate $\pi(\calZ|\bfX_n)$ using Laplace approximation in Section \ref{laplace}.

\subsection{Laplace Approximation}\label{laplace}
In this section, we approximate $\pi(\calZ|\bfX_n)$ in (11) using Laplace approximation. Suppose $\bsSig^*_{\calZ}=(\sigma_{\calZ,ij}^*)\in \calU_{\calZ}(\tau)$ is the minimizer of $r_{\calZ}(\bsSig_{\calZ},\bfX_n)$ on $\calU_{\calZ}(\tau)$, where $\bsSig_{\calZ}=(\sigma_{\calZ,ij})\in\calU_{\calZ}(\tau)$ and $\bfX_n$ is given. Then, the Laplace approximation is done around $\bsSig^*_{\calZ}$. So we have to find $\bsSig^*_{\calZ}$ first. Observe that $\bsSig^*_{\calZ}$ is the solution of optimization problem (12). 
\begin{align}
\minimize_{\bsSig_{\calZ}=(\sigma_{\calZ,ij})\in \calU_{\calZ}(\tau)} \log|\bsSig_{\calZ}|+\trace\left(\bfS\bsSig_{\calZ}^{-1}\right)+\frac{1}{nv^2}\sum_{\zij=1}\sigma_{\calZ,ij}^2+\frac{\lambda}{n}\sum_{i=1}^p \sigma_{\calZ,ii}
\end{align}
Objective function $r_{\calZ}(\bsSig_{\calZ},\bfX_n)$ in (12) can be seen as a regularized negative log likelihood by putting $\ell_2$-type penalty on off-diagonal entries and $\ell_1$-type penalty on diagonal entries. Note that $r_{\calZ}(\cdot,\bfX_n)$ is not convex on $\calU_{\calZ}(\tau)$ and $\bsSig^*_{\calZ}$ cannot be obtained in the closed form. Hence, the optimization problem (12) can be reduced to convex optimization problem if we consider the set $\calQ_{\calZ}(\tau)$ instead of $\calU_{\calZ}(\tau)$ in the optimization problem (12).

 Observe that $r_{\calZ}(\cdot,\bfX_n)$ is convex on the set $\calQ_{\calZ}(\tau)=\{W\in \calU_{\calZ}(\tau): W\prec 2\bfS\}$. Provided that $\bsSig_{\calZ}^*$ belongs to the set $\calQ_{\calZ}(\tau)$, solving the reduced optimization problem is more desired, since the convexity makes it easier to deriving algorithm for finding the solution $\bsSig_{\calZ}^*$ of the optimization problem (12). We have to resort to some numerical algorithm, as $\bsSig_{\calZ}^*$ cannot be obtained in the closed form. If such algorithm converges to the stationary point of the reduced optimization problem, we obtain a local minimum of $r_{\calZ}(\cdot,\bfX_n)$ but the local minimum can be global minimum due to the convexity.

 But if $p$ depends on $n$, $\bsSig^*_{\calZ}\prec 2\bfS$ does not necessarily hold. Note that we assume the dependency between $p$ and $n$ to derive asymptotic statistical properties for our proposed method in this paper, which are to be discussed in Section \ref{errorlap}. Thus, we pose assumptions on parameters so that $\bsSig^*_{\calZ}\prec 2\bfS$. Consider the following assumptions : 
 \begin{enumerate}
\item[\bf (A1)] $p \asymp n^{\beta}$ for some constant $0<\beta<1$.  
\item[\bf (A2)] $0<v$ is some constant, $1<\tau$, $\tau=\bigO(1)$, and $\lambda=\bigO(1)$.
\end{enumerate}

Assuming (A1) and (A2), we show that $\bsSig_{\calZ}^*\prec 2\bfS$ with probability tending to one. For simplicity, we may consider when $\calZ=\mathbf{1}_{p(p-1)/2}$ , since the similar argument can be applied to general $\calZ$, where $\mathbf{1}_{p(p-1)/2}$ is a $p(p-1)/2-$vector with all entries being $1$. Then, $\bsSig^*_{\calZ}$ must satisfy the normal equation 
\begin{align*}
    -\bsSig_{\calZ}^{-1}+\bsSig_{\calZ}^{-1}\bfS\bsSig_{\calZ}^{-1}+\bsLambda\circ\bsSig_{\calZ}+\frac{\lambda}{n}\mathbf{I}_p=\bfzero_{p\times p},
\end{align*}
or equivalently, 
\begin{align}
       -\bsSig_{\calZ}+\bsSig_{\calZ}\bfS\bsSig_{\calZ}+\bsSig_{\calZ}\left(\bsLambda\circ\bsSig_{\calZ}\right)\bsSig_{\calZ}+\frac{\lambda}{n}\bsSig_{\calZ}^2=\bfzero_{p\times p}, 
\end{align}
where $\bsLambda\in\calM$ is the matrix with 0 on diagonal entries and $1/nv^2$ on off-diagonal entries and $\mathbf{I}_p$ and $\bfzero_{p\times p}\in\calM$ is identity matrix and zero matrix, respectively. Since $\lambda=\bigO(1)$ and $\tau=\bigO(1)$, 
\begin{align}
\begin{split}
    ||\frac{\lambda}{n}\bsSig_{\calZ}^2||_2&\leq \frac{\lambda}{n}||\bsSig_{\calZ}||_2^2\\
    &\leq \frac{\lambda}{n}\tau^2 \rightarrow 0,
\end{split}
\end{align}
as $n\rightarrow\infty$. Also, 
\begin{align}
\begin{split}
    ||\bsSig_{\calZ}\left(\bsLambda\circ\bsSig_{\calZ}\right)\bsSig_{\calZ}||_2&\leq ||\bsSig_{\calZ}||_2^2 ||\bsLambda\circ\bsSig_{\calZ}||_2\\
    &\leq p\tau^2 ||\bsLambda\circ\bsSig_{\calZ}||_{\infty}\\
    &\leq p\tau^2 \max_{\zij=1}\{1/nv^2\cdot|\sigma_{\calZ,ij}|\}\\
    &\leq \tau^3 \frac{p}{nv^2}\rightarrow 0,
\end{split}
\end{align}
as $n\rightarrow \infty$. Here we used (1) in the second inequality and note that the last inequality holds by (A1) and $\tau=\bigO(1)$. Since the equation (13) is equivalent to equation (16)
\begin{align}
    2\bfS-\bsSig_{\calZ}=\bfS-[\bsSig_{\calZ}\left(\bsLambda\circ\bsSig_{\calZ}\right)\bsSig_{\calZ}+\frac{\lambda}{n}\bsSig_{\calZ}^2],
\end{align}
(14) and (15) imply that $\bsSig_{\calZ}^*$ belongs to $\calQ_{\calZ}(\tau)$ for all sufficiently large $n$. 

Thus, assuming (A1) and (A2), the optimization problem (12) can be reduced to the following optimization problem (17) 
\begin{align}
\minimize_{\bsSig_{\calZ}=(\sigma_{\calZ,ij})\in \calQ_{\calZ}(\tau)} \log|\bsSig_{\calZ}|+\trace\left(\bfS\bsSig_{\calZ}^{-1}\right)+\frac{1}{nv^2}\sum_{\zij=1}\sigma_{\calZ,ij}^2+\frac{\lambda}{n}\sum_{i=1}^p \sigma_{\calZ,ii},
\end{align}
where $\calQ_{\calZ}(\tau)=\{\calG\in\calU_{\calZ}(\tau):\calG\prec2\bfS\}$, which is a convex optimization problem. We consider the optimization problem (17) instead of (12). 

Now it remains to solve (17). To solve (17) as to obtain $\bsSig_{\calZ}^*$, we resort to a numerical algorithm that solves (17) as $\bsSig_{\calZ}^*$ cannot be obtained in the closed form. We provide a block coordinate descent algorithm for solving (17) in Section \ref{bcd}.  

Provided that $\bsSig_{\calZ}^*$ is in the hand, we apply Laplace approximation to (11). Let $\bsSig_{\calZ}=\bsSig^*_{\calZ}+\bsDelta_{\calZ}$, where $\bsDelta_{\calZ}=(\deltaij)$. Then we have 
\begin{align}
    r_{\calZ}(\bsSig_{\calZ},\bfX_n)=r_{\calZ}(\bsSig^*_{\calZ},\bfX_n)+k_{\calZ}\left(\bsDelta_{\calZ},\bfX_n\right)-\log\left|\bsSig_{\calZ}^*\right|-\trace\left(\bfS\bsOmega_{\calZ}^*\right),
\end{align}
where $k_{\calZ}\left(\bsDelta_{\calZ},\bfX_n\right)=\log\left|\bsSig^*_{\calZ}+\bsDelta_{\calZ}\right|+\trace\left(\bfS\left(\bsSig^*_{\calZ}+\bsDelta_{\calZ}\right)^{-1}\right)+\frac{1}{nv^2}\sum_{\zij=1}(2\sgij^*\deltaij+\deltaij^2)+\frac{\lambda}{n}\sum_{i=1}^p\deltaii$ and  $\bsOmega_{\calZ}^*=\left(\bsSig^*_{\calZ}\right)^{-1}=(\omegaij^*)$. Substituting (18) into (11), 
\begin{align}
    \pi\left(\calZ|\bfX_n\right)&\propto  \left(\frac{\pi}{(1-\pi)v}\frac{1}{\sqrt{2\pi}}\right)^{\# \calZ}\exp\left(-\frac{n}{2}r_{\calZ}\left(\bsSig_{\calZ}^*,\bfX_n\right)\right)\left|\bsSig_{\calZ}^*\right|^{\frac{n}{2}}\exp\left(\frac{n}{2}\trace\left(\bfS\bsOmega_{\calZ}^*\right)\right) \nonumber \\
    &\quad \int_{\bsSig_{\calZ}^*+\bsDelta_{\calZ} \in\calU_{\calZ}(\tau)}\exp\left(-\frac{n}{2}k_{\calZ}\left(\bsDelta_{\calZ},\bfX_n\right)\right)d\bsDelta_{\calZ},
\end{align}
where $d \bsDelta_{\calZ}=\prod_{(i,j)\in\bar{E}_{\calZ}}d \Delta_{\calZ,ij}$. Note that $k_{\calZ}$ is uniquely minimized at $\bsDelta_{\calZ}=\bfzero_{p\times p}$. Hence we have $(p+\#\calZ)\times(p+\#\calZ)$ Hessian matrix  $\hessian_{\bsSig_\calZ^*}=(h_{\{(i,j),(l,m)\}})$ of $k_{\calZ}$ at $\bsDelta_{\calZ}=\bfzero$, where $(i,j), (l,m)\in \bar{\edge}_{\calZ}$, and apply Laplace approximation to the integral in (19). Let $\bfU_{\calZ}=\bsOmega_{\calZ}^*\bfS\bsOmega_{\calZ}^*=(u_{\calZ,ij})$. Suppressing dependency on $\calZ$, write $u_{\calZ,{ij}}$ as $u_{ij}$ and $\omega_{\calZ,ij}^*$ as $\omega_{ij}$ for simplicity. By the supplementary material, we see that 
\begin{align}
 h_{\{(i,j),(l,m)\}}=\begin{cases}
2(-\omega_{ii}^*\omega_{jj}^*-(\omega_{ij}^*)^2+2u_{ji}\omega_{ji}^*+u_{jj}\omega_{ii}^*+u_{ii}\omega_{jj}^*+\frac{1}{nv^2})&\text{ , } i<j, \text{ } l<m, \text{ } (i,j)=(l,m)\\
 2(-\omega_{il}^*\omega_{jm}^*-\omega_{im}^*\omega_{jl}^*+u_{jl}\omega_{mi}^*+u_{jm}\omega_{li}^*+u_{li}\omega_{jm}^*+u_{mi}\omega_{jl}^*)&\text{ , } i<j, \text{ } l<m, \text{ } (i,j)\neq(l,m)\\
 2(-\omega_{il}^*\omega_{jl}^*+u_{jl}\omega_{li}^*+u_{li}\omega_{jl}^*)&\text{ , } i<j, \text{ } l=m \\
 -(\omega_{il}^*)^2+2u_{il}\omega_{li}^*&\text{ , } i=j, \text{ } l=m 
 \end{cases}.
\end{align}
 By (20), we approximate $\pi\left(\calZ|\bfX_n\right)$ by $\pi^*\left(\calZ|\bfX_n\right)$ as follows.
\begin{align}
\begin{split}
\pi^*\left(\calZ|\bfX_n\right)&\propto \left(\frac{\pi}{(1-\pi)v}\frac{1}{\sqrt{2\pi}}\right)^{\# \calZ}\exp\left(-\frac{n}{2}r_{\calZ}(\bsSig_{\calZ}^*,\bfX_n)\right)\left|\bsSig_{\calZ}^*\right|^{\frac{n}{2}}\exp\left(\frac{n}{2}\trace\left(\bfS\bsOmega_{\calZ}^*\right)\right)\\
&\qquad \exp\left(-\frac{n}{2}k_{\calZ}(\bfzero_{p\times p},\bfX_n)\right)\left(\frac{4\pi}{n}\right)^{(p+\#\calZ)/2}\left|\hessian_{\bsSig_{\calZ}^*}\right|^{-\frac{1}{2}}\\
&=\left(\frac{\pi}{(1-\pi)v}\frac{1}{\sqrt{2\pi}}\right)^{\# \calZ}\exp\left(-\frac{n}{2}r_{\calZ}(\bsSig_{\calZ}^*,\bfX_n)\right)\left(\frac{4\pi}{n}\right)^{(p+\#\calZ)/2}\left|\hessian_{\bsSig_{\calZ}^*}\right|^{-\frac{1}{2}}.
\end{split}
\end{align}

 Using this approximated posterior model probability of $\calZ$, which is a graphical structure indicator, we use Metropolis-Hastings algorithm proposed by \cite{liu2019empirical} to generate MCMC samples of $\calZ$. \cite{liu2019empirical} considered a symmetric proposal distribution $u$, which samples $\calZ'$ uniformly from $\calZ$ that differs from $\calZ$ in only one entry. The specific description of $u$ can be found in Section 5.1 of \cite{liu2019empirical}. This gives the following Metropolis-Hastings algorithm to generate MCMC samples of $\calZ$.  
\begin{algorithm}
    \caption{Metropolis-Hastings algorithm for generating MCMC samples of $\calZ$}
    \label{algorithm1}
    \begin{algorithmic}[1]
    \State Initial model structure indicator : $\calZ^{(0)}$, Given data : $\bfX_n$
    \For{$i=0,1,\dots,k-1$}
    \State $\mathcal{Z}^{\text{cand}}\sim u(\mathcal{Z}|\mathcal{Z}^{(i)})$
    \State $\alpha_i=\min\{1,\frac{\pi^*(\calZ^{\text{cand}}|\bfX_n)}{\pi^*(\calZ^{(i)}|\bfX_n)}\}$
    \State $U_i\sim U(0,1)$
    \State If $U_i\leq \alpha_i$, $\calZ^{(i+1)}=\calZ^{\text{cand}}$. Else, $\calZ^{(i+1)}=\calZ^{(i)}$. 
    \EndFor
    \end{algorithmic}
\end{algorithm}

Here $\pi^*(\calZ|\bfX_n)$ is defined in (21). We choose the final model by either MPM or MAP. Suppose $\tilde{\calZ}$ is chosen as the final model. With $\tilde{\calZ}$ and $\bfX_n$ given, by Bayes' rule, the conditional posterior of $\bsSig$ is given as following:
\begin{align*}
    \pi\left(\bsSig|\bfX_n,\tilde{\calZ}\right)&\propto \pi\left(\bfX_n|\bsSig,\tilde{\calZ}\right)\pi\left(\bsSig|\tilde{\calZ}\right)\pi^u\left(\tilde{\calZ}\right) \nonumber \\
&\propto\exp\left(-\frac{n}{2}\log|\bsSig|-\frac{n}{2}\trace\left(\bfS\bsSig^{-1}\right)-\frac{n}{2}p(\bsSig,\tilde{\calZ})\right)\dsone\left(\bsSig\in\calU_{\tilde{\calZ}}(\tau)\right)\nonumber \\
&=\exp\left(-\frac{n}{2}r_{\tilde{\calZ}}(\bsSig,\bfX_n)\right)\dsone\left(\bsSig\in\calU_{\tilde{\calZ}}(\tau)\right).
\end{align*}
Here $p(\cdot,\cdot)$ and $r_{\calZ}(\cdot,\cdot)$ are defined in (9) and (10). In this paper, we estimate $\bsSig$ by the mode of $\pi\left(\bsSig|\bfX_n,\tilde{\calZ}\right)$. Note that the mode of $\pi\left(\bsSig|\bfX_n,\tilde{\calZ}\right)$ is equivalent to the solution of the optimization problem (17) with $\calZ=\tilde{\calZ}$. Thus, the mode of $\pi\left(\bsSig|\bfX_n,\tilde{\calZ}\right)$ can be found using block coordinate descent algorithm provided in Section \ref{bcd}.

\subsection{Error by Laplace Approximation}\label{errorlap}
 In this section, we show that the error by Laplace approximation becomes asymptotically marginal as in (21) under regular conditions. Note that this can established if we show that
 \begin{align}
     \pi(\calZ|\bfX_n)/\pi^*(\calZ|\bfX_n)\rightarrow 1 
 \end{align}
with probability tending to one in $\bbP_0-$probability.
 
 Denote the true covariance matrix by $\bsSig_0$. Suppose that the number of nonzero off-diagonal entries in $\bsSig_0$ is controlled by a positive integer $s_0$, where $0<s_0<\binom{p}{2}/2$. For a covariance matrix $\bsSig=(\sgij)\in \calM^+$, let $s(\bsSig)$ be the number of edges in the covariance graph induced by $\bsSig$. Define the set 
\begin{align*}
    \calU(s_0,\tau_0)=\{\bsSig\in\calM^+: s(\bsSig)\leq s_0,1/\tau_0\leq \eigmin\left(\bsSig\right)\leq \eigmax \left(\bsSig\right)\leq \tau_0 \},
\end{align*}
\noindent where $\tau_0>1$ is some constant. With these notations, in addition to the assumptions (A1) and (A2), we pose additional assumptions on parameters and $\bsSig_0$ :
\begin{enumerate}
\item[\bf (A3)] $\bsSig_0 \in \calU(s_0,\tau_0)$.
\item[\bf (A4)] $0<\tau_0$ is some constant, $\tau_0<\tau$, $q\in(0,1)$, and $q\asymp \frac{\log p}{p^2}$.
\end{enumerate}

To explain (A3), the assumption (A3) was first considered in \cite{bickel2008covariance} and  often used in statistical inference of sparse covariances for either Frequentist literature or Bayesian literature. The positive integer $s_0$ controls the sparsity of true covariance matrix $\bsSig_0$ and the upper bound on eigenvalue of $\bsSig_0$, $\tau_0$, together with $\tau$, is often used in converting $||\bsSig_0^{-1}-\bsSig^{-1}||_{\frob}$ to $||\bsSig_0-\bsSig||_{\frob}$. (A4) is another assumption on parameters to attain posterior convergence rate of covariance matrix under the Frobenius norm, which turns out to be $\epsilon_n=\sqrt{(p+s_0)\log p/n}$. 
 
 Note that the result of posterior convergence rate is necessary to establish $||\bsSig_{\calZ}^*-\bsSig_{\calZ}||_{\frob}=\bigO_p(\eta_n)$, which is crucial to prove (22), where $\eta_n=\sqrt{p}\epsilon_n$. This implies that the posterior and true covariance matrix $\bsSig_0$ are concentrated around the projection of true covariance matrix $\bsSig_0$ onto the model $\calZ$ at rate $\eta_n$ similar to remark of (3.15) in \cite{banerjee2015bayesian}. A similar argument was also made in \cite{banerjee2015bayesian}. Using auxiliary results in the supplementary material, the error due to Laplace approximation tends to zero with probability tending to one if $(p+\#\calZ)^2\eta_n=o(1)$ under regular conditions. This gives Theorem \ref{thm:error_rate}.
 
 \begin{theorem}\label{thm:error_rate}
Assume that $0<\beta<1/2$ in (A1), (A2)-(A4). Also, suppose that $\tau^4\leq p$, $\max\{1/\tau,1/p\}<\lambda<\log p/\tau_0$, $\tau>3$, $\tau^2\tau_0^2\leq s_0\log p$, and $n\geq s_0\log p/[(1-\tau_0/\tau)^2\tau^4]$. Further suppose that $(p+\#\calZ)^2\eta_n=o(1)$, where $\eta_n=\sqrt{p}\epsilon_n$ and $\epsilon_n=\sqrt{(p+s_0)\log p/n}$.  Then, under prior (7),
\begin{align*}
    \pi(\calZ|\bfX_n)/\pi^*(\calZ|\bfX_n)\rightarrow 1
\end{align*}
in $\bbP_0-$ probability, which implies that the error by Laplace approximation in (21) tends to zero in $\bbP_0-$probability.
\end{theorem}
 
The proof is provided in Appendix \ref{appendixA}. The proof uses the techniques considered by \cite{rothman2008} and \cite{banerjee2015bayesian}. Since the proof uses the result of posterior convergence rate, it remains to show that $\epsilon_n=\sqrt{(p+s_0)\log p/n}$ is the posterior convergence rate of the covariance matrix under the Frobenius norm. Under prior (7), this follows from Theorem \ref{thm:conv_rate}.

\begin{theorem}\label{thm:conv_rate}
Let $\mathbf{X}_n=(X_1,\dots,X_n)^t$ be the random sample from $\N_p(\bfzero,\bsSig)$ and consider the prior (7). Assume (A1)-(A4) and suppose that $\tau^4\leq p$, $\max\{1/\tau,1/p\}<\lambda<\log p/\tau_0$, $\tau>3$, $\tau^2\tau_0^2\leq s_0\log p$, and $n\geq s_0\log p/[(1-\tau_0/\tau)^2\tau^4]$. If $\epsilon_n=o(1)$, 
\begin{align*}
\pi(||\bsSig-\bsSig_0||_{\frob}\geq M\epsilon_n|\mathbf{X}_n)\rightarrow 0, 
\end{align*}
for some constant $M>0$ in $\bbP_0$-probability.
\end{theorem}

\subsection{Block Coordinate Descent Algorithm}\label{bcd}
In this section, we propose a block coordinate descent algorithm that solves (17). For convenience, write $\bsSig_{\calZ}$ as $\bsSig$ suppressing dependency on $\calZ$.  Partition $\bsSig$ and $\bfS$ as follows.
\begin{align}
    \bsSig=\left(\begin{array}{cc}
        \bsSig_{11} & \bssig_{12} \\
         \bssig_{12}^t & \sigma_{22} 
    \end{array}
    \right), \quad \bfS=\left(\begin{array}{cc}
    \bfS_{11} & \bss_{12} \\
    \bss_{12}^t & s_{22}
    \end{array}\right),
\end{align}
where $\bsSig_{11}$ $(\bfS_{11})$ is a $(p-1)\times (p-1)$ matrix and $\sigma_{22}$ $(s_{22})$ is scalar. Let $\bsbeta=\bssig_{12}$ and $\bsgamma=\sigma_{22}-\bssig_{12}^t\bsSig_{11}^{-1}\bssig_{12}$. Note that by the constraint $\calQ_{\calZ}(\tau)$, some entries of $\bsbeta$ are fixed as 0. For simplicity, write $\bsbeta=(\bsbeta_1^t,\bsbeta_0^t)^t$, where $\bsbeta_1$ $(\bsbeta_0)$ is vector with entries being $\sgij$ corresponding to $\zij=1$ $(\zij=0)$, by rearranging rows/columns of $\bsSig$ and $\bfS$ in (23). Since $\bsbeta_0$ should be fixed due to $\calZ$, we update only $\bsbeta_1$ in $\bsbeta$. With fixed $\bsSig_{11}$, neglecting terms that do not depend on $\bsbeta_1$ or $\bsgamma$ in (17), we have 
\begin{align}
    \minimize_{(\bsbeta_1,\bsgamma)} \log \bsgamma+\frac{\bsbeta_1^t[\bsSig_{11}^{-1}\bfS_{11}\bsSig_{11}^{-1}]^1\bsbeta_1-2\bsbeta_1^t[\bsSig_{11}^{-1}\bfs_{12}]^1+s_{22}}{\bsgamma}+\bsbeta_1^t\bsTheta\bsbeta_1+\frac{\lambda}{n}(\bsgamma+\bsbeta_1^t[\bsSig_{11}^{-1}]^1\bsbeta_1),
\end{align}
where $[\bsSig_{11}^{-1}\bfS_{11}\bsSig_{11}^{-1}]^1$ and $[\bsSig_{11}^{-1}]^1$ denote some principal minor matrix of $\bsSig_{11}^{-1}\bfS_{11}\bsSig_{11}^{-1}$ and $\bsSig_{11}^{-1}]$, respectively, $[\bsSig_{11}^{-1}\bfs_{12}]^1$ is subvector of $\bsSig_{11}^{-1}\bfs_{12}$, and $\bsTheta=\diag(\frac{1}{nv^2})$. Let $f(\bsbeta_1,\bsgamma)$ be the objective function in (24). One can see that $f(\bsbeta_1,\bsgamma)$ is quite similar to (4) in \cite{wang14}, except for the term $\bsbeta_1^t\bsTheta\bsbeta_1$. Neglecting terms that do not depend on $\bsgamma$ in (24), (24) is reduced to 
\begin{align}
    \minimize_{\bsgamma} \log \bsgamma+\frac{\bsu}{\bsgamma}+\frac{\lambda}{n}\bsgamma,
\end{align}
where $\bsu=\bsbeta_1^t[\bsSig_{11}^{-1}\bfS_{11}\bsSig_{11}^{-1}]^1\bsbeta_1-2\bsbeta_1^t[\bsSig_{11}^{-1}\bfs_{12}]^1+s_{22}$. Note that $\bsu>0$ with probability tending to one, which can be shown in Proposition \ref{positivity}. 
\begin{prop}\label{positivity}
Let $\bfu=\bsbeta_1^t[\bsSig_{11}^{-1}\bfS_{11}\bsSig_{11}^{-1}]^1\bsbeta_1-2\bsbeta_1^t[\bsSig_{11}^{-1}\bfs_{12}]^1+s_{22}$ as in (24). Then $\bfu>0$ with probability tending to one. 
\end{prop}
Substituting $\rho=\lambda/n$ and $a=\bfu$ in (5) of \cite{wang14}, one can see that the solution of (25), denoted by $\hat{\bsgamma}$, is 
\begin{align}
    \hat{\bsgamma}=\frac{-1+\sqrt{1+4\bfu\rho}}{2\rho}.
\end{align}
Note that the positivity of $\bfu$ ensures the existence of (26). Now let $l(\bsxi)=f(\bsxi,\hat{\bsgamma})$. Note $\hat{\bsgamma}$ depends on $\bsbeta_1$ and $\bsxi$ is a coordinate independent of $\bsbeta_!$. Neglecting terms that do depend on $\bsxi$ in $l$, we have
\begin{align*}
    l(\bsxi)=\bsxi^t\left(\bsTheta+\frac{\lambda}{n}[\bsSig_{11}^{-1}]^1+[\bsSig_{11}^{-1}\bfS_{11}\bsSig_{11}^{-1}]^1/\hat{\bsgamma}\right)\bsxi-2\bsxi^t[\bsSig^{-1}\bfs_{12}]^1/\hat{\bm{\gamma}}.
\end{align*}

\noindent Since $l$ is convex, the minimizer $\hat{\bsxi}$ can be found by solving $\frac{\partial h}{\partial \bsxi}=\bfzero$. So,
\begin{align}
    \frac{\partial l}{\partial \bsxi}&=2\left(\bsTheta+\frac{\lambda}{n}[\bsSig_{11}^{-1}]^1+[\bsSig_{11}^{-1}\bfS_{11}\bsSig_{11}^{-1}]^1/\hat{\bsgamma}\right)\bsxi-2[\mathbf{\Sigma}_{11}^{-1}\mathbf{s}_{12}]^1/\hat{\bsgamma}=\bfzero \nonumber \\
    &\Rightarrow \hat{\bsxi}=\left(\bsTheta+\frac{\lambda}{n}[\bsSig_{11}^{-1}]^1+[\bsSig_{11}^{-1}\bfS_{11}\bsSig_{11}^{-1}]^1/\hat{\bsgamma}\right)^{-1}[\mathbf{\Sigma}_{11}^{-1}\mathbf{s}_{12}]^1/\hat{\bm{\gamma}}.
\end{align}

Write $\bssig_{12}=((\bssig_{12})_1^t,(\bssig_{12})_0^t)^t$ as for $\bsbeta$. From (26) and (27), $\bssig_{12}$ and $\sigma_{22}$ can be updated by $\hat{\bssig}_{12}$ and $\hat{\sigma}_{22}$ respectively as follows: 
\begin{align}
    &(\hat{\bssig}_{12})_1=\left(\bsTheta+\frac{\lambda}{n}[\bsSig_{11}^{-1}]^1+[\bsSig_{11}^{-1}\bfS_{11}\bsSig_{11}^{-1}]^1/\hat{\bsgamma}\right)^{-1}[\mathbf{\Sigma}_{11}^{-1}\mathbf{s}_{12}]^1/\hat{\bm{\gamma}},\\
    &(\hat{\bssig}_{12})_0=\bfzero,\\
    &\hat{\sigma}_{22}=\hat{\bsgamma}+\hat{\bssig}_{12}^t\bsSig_{11}^{-1}\hat{\bssig}_{12}.
\end{align}

By (28)-(30), we propose a block coordinate descent algorithm for solving (17) as in Algorithm \ref{algorithm2}. 

Since we're estimating covariance matrices, the estimated matrix resulted from the given block coordinate descent algorithm should be positive definite. Also, we may wish the given algorithm converges to a stationary point of (17). This is because by the discussion in Section \ref{laplace}, if the objective function in (17) attains the stationary point, then the function is convex at such stationary point and so the stationary point becomes a local minimum point of the function. It is possible that a local minimum point becomes a global minimum point for sufficiently large $n$ by (13)-(16) and (A1)-(A2). We verify this desired property in Proposition \ref{convergence}.

\begin{prop}\label{convergence}
Updating a single row/column by the rule described in Algorithm \ref{algorithm2} results in a positive definite matrix. Furthermore, the proposed algorithm converges to a stationary point of the objective function in (17). 
\end{prop}

\newpage
 \begin{algorithm}
    \caption{Block Coordinate Descent algorithm for solving (17)}
    \label{algorithm2}
     \begin{algorithmic}[1]
    \State Fix $\epsilon>0$ and initialize $\mathbf{\Sigma}^{(0)}=\diag(\mathbf{S})+\frac{\lambda}{n}\I_p$ 
    \For{$i=1,\dots,k$}
    \State $\mathbf{\Sigma}^{(i)}=\mathbf{\Sigma}^{(i-1)}$
    \For{$j=1,2,\dots,p$}
    \State $\mathbf{\Sigma}^{(i)}_j$ : Rearrange rows/columns in $\mathbf{\Sigma}^{(i)}$  so that $j$th diagonal entry of  $\mathbf{\Sigma}^{(i)}$ is placed on the last diagonal entry.
    \State Partition $\mathbf{\Sigma}^{(i)}_j$ and $\mathbf{S}$ as in (23).  
    \State Calculate $\hat{\bm{\gamma}}$ as in (26).
    \State Update $\bm{\sigma}_{12}$ and $\bm{\sigma}_{12}^t$ as in (28)-(29). 
    \State Update $\sigma_{22}$ as in (30). 
    \State Rearrange $\mathbf{\Sigma}^{(i)}_j$ so that the last diagonal entry is on $j$th diagonal entry and set $\mathbf{\Sigma}^{(i)}$=$\mathbf{\Sigma}^{(i)}_j$.
    \EndFor
    \If {$||\mathbf{\Sigma}^{(i)}-\mathbf{\Sigma}^{(i-1)}||_{\text{F}}<\epsilon$} 
    \State \Return $\mathbf{\Sigma}^{(i)}$.
    \EndIf
    \EndFor
    \end{algorithmic}
\end{algorithm}

\section{Simulation}\label{sec:simul}
\subsection{Numerical Study}
To evaluate the performance of the proposed estimator in this paper, we perform a simulation study. We consider following five covariance models with covariance matrix $\bsSig=(\sgij)$ or its inverse $\bsOmega=(\omega_{ij})$.

\begin{itemize}[leftmargin=*]
    \item \texttt{Model 1.} Random Structure \rom{1}: $\sgij=\sigma_{ji}$ is non-zero with probability 0.02 for $i<j$ independently of other off-diagonal entries. For a nonzero $\sgij=\sigma_{ji}$, assign $1$ or $-1$ randomly. Diagonal entries are chosen as a constant so that the condition number of $\bsSig$ under $||\cdot||_2$ is approximately $p$. 
    \item \texttt{Model 2.} Random Structure \rom{2}: $\bsSig=(\bfB+\delta\mathbf{I}_p)/(1+\delta)$, where $\bfB=(b_{ij})\in\calM$ with $b_{ii}=1$ for $i=1,2,\dots,p$ and $b_{ij}=0.5\times \text{Ber}(0.2)$ for $1\leq i<j\leq p$ , and $\mathbf{I}_p$ is a $p\times p$ identity matrix. $\delta=\max\{-\eigmin(\bfB),0\}+0.05$ so that $\bfB+\delta\mathbf{I}_p\in \calM^+$ and each entry of $\bfB+\delta\mathbf{I}_p$ is divided by $1+\delta$ to normalize diagonal entries.  
    \item \texttt{Model 3.} First-order Moving Average Model: For $1\leq i<j\leq p$, $\sgij=\sigma_{ji}=0.4$ if $j=i+1$ and $0$ otherwise. Diagonal entries are chosen as a constant so that the condition number of $\bsSig$ by $||\cdot||_2$ is approximately $p$.
    \item \texttt{Model 4.} Second-order Moving Average Model: $\sgii=1$. For $1\leq i<j\leq p$, $\sgij=\sigma_{ji}=0.5$ if $j=i+1$, $\sgij=\sigma_{ji}=0.25$ if $j=i+2$, and $\sgij=\sigma_{ji}=0$ otherwise. 
    \item \texttt{Model 5.} Inverse of Toeplitz matrix: $\omega_{ij}=0.75^{|i-j|}$.
\end{itemize}

\texttt{Model 1} and \texttt{Model 3} were considered by \cite{bien2011sparse}. Choice of diagonal entries ensures the positive definiteness of sample covariance matrix when $n>p$ as in \cite{bien2011sparse} and \cite{rothman2008}. \texttt{Model 2} is similar to models considered by \cite{fang2015}, \cite{cai2011}, and \cite{rothman2008}. Note that each covariance matrix is sparse.

For each model, we consider $p=50$ $(n=100)$ and $p=100$ $(n=200)$. We generate random sample sizes of $n$ from $p$-dimensional Gaussian distribution with mean $\bfzero_p$ and covariance in each model and run 100 replications for each model. We measure the performance of the estimator by specificity (\texttt{sp}), sensitivity (\texttt{se}), root mean square error (\texttt{rmse}), $||\hat{\bsSig}-\bsSig||_{\frob}/p$, max norm (\texttt{mnorm}), $||\hat{\bsSig}-\bsSig||_{\infty}$, and spectral norm (\texttt{2norm}), $||\hat{\bsSig}-\bsSig||_2$. Specificity and sensitivity denote the ratio of correctly estimated non-zero off-diagonal entries over total non-zero off-diagonal entries in the estimated covariance matrix and the ratio of correctly estimated zero off-diagonal entries over total zero off-diagonal entries in the estimated covariance matrix, respectively. We claim that the estimated off-diagonal entry is $0$ if its absolute value does not exceed $0.001$ following \cite{Fan2009} and \cite{wang12}. For each measure, we provide the mean and the standard deviation over 100 replications. 

Competing methods in the simulation study are covariance graphical lasso (GL) proposed by  \cite{bien2011sparse} and sample covariance matrix (Samp). For our proposed method, we obtain MCMC samples of $\calZ$ from the posterior using Algorithm \ref{algorithm1} described in Section \ref{laplace} and choose the final model either by MPM or MAP. We generate 12000 posterior samples after 3000 burn-in for our proposed method. Given the chosen model, we estimate $\bsSig$ by solving (17) using Algorithm \ref{algorithm2} by the argument in Section \ref{laplace}. The simulation results are provided in Table \ref{tb:sim1} and Table \ref{tb:sim2}. Note that for each measure,  the boldfaced value denotes the best performance compared to other competing methods in each dimension (sample size) and model. 

We discuss the results in Table \ref{tb:sim1} first. Our proposed method tends to perform better than GL and Samp in terms of \texttt{sp} for all models regardless of model choice MPM or MAP for the proposed method. This demonstrates that the prior we considered in this paper is indeed effective in introducing sparsity to the covariance structures. Note that Samp performs poorly in terms of both \texttt{se} and \texttt{sp} in most cases, especially \texttt{sp}. Samp showed extremely small \texttt{sp}. Though better than Samp, GL also showed poor \texttt{sp}. Note that all the methods we have considered in the simulation study showed good performance in terms of \texttt{se} for \texttt{Model 1, Model 3,} and \texttt{Model 5} with $p=50$ $(n=100)$. But in general, GL and Samp outperformed our proposed method in terms of \texttt{se}. However, since both GL and Samp are extremely ineffective in introducing sparsity to the model, which can be seen from the results for \texttt{sp} in each model, if \texttt{se} for the proposed method is not significantly lower than that for other methods, our proposed method may be more useful than other methods in the inference of sparse covariance structures.

Now we discuss the results in Table \ref{tb:sim2}. In terms of \texttt{rmse} and \texttt{2norm}, our proposed method outperforms other methods in most of the models except for \texttt{Model 5} with $p=100$ $(n=200)$. Also, in terms of \texttt{mnorm}, Samp tends to perform better than our proposed method and graphical lasso.

\begin{sidewaystable}
\centering
 \caption{\texttt{\small sp} and \texttt{\small se} under the covariance structures in  \texttt{\small Model 1, 2, 3, 4}, and \texttt{\small 5}.}
 \label{tb:sim1}
 \scalebox{0.78}{
 \begin{tabular}{cccccc cccc}
		\hline\hline \\ [-1.7em] 
		\multirowcell{2}{}& \multirowcell{2}{Measure} & \multicolumn{4}{c}{p=50\, (n=100)} &\multicolumn{4}{c}{p=100\, (n=200)}\\ [-1.7em]\\\cmidrule(lr){3-6}\cmidrule(lr){7-10}  \\ [-1.7em]
		& & Proposed (MPM) & Proposed (MAP) & GL & Samp & Proposed (MPM) & Proposed (MAP) & GL &  Samp\\  [-1.7em]\\ \hline\hline 
		\\[-1.2em]
		\multirowcell{2}{\texttt{Model 1}}& \texttt{sp} & \textbf{0.998 (0.014)} & 0.981 (0.017) & 0.583 (0.027) & 0.003 (0.002) & \textbf{0.999 (0.010)} & \textbf{0.999 (0.009)} & 0.900 (0.008) & 0.004 (0.009)  \\ 
		&\texttt{se}& \textbf{1.000 (0.005)} & 0.982 (0.005) & \textbf{1.000 (0.004)} & \textbf{1.000 (0.000)} & 0.989 (0.011) & \textbf{1.000 (0.002)} & \textbf{1.000 (0.002)} & \textbf{1.000 (0.000)} \\[-1.2em]\\ \hline
			\\[-1.2em] 
		\multirowcell{2}{\texttt{Model 2}}& \texttt{sp}& \textbf{0.986 (0.004)} & 0.984 (0.005) & 0.268 (0.009) & 0.003 (0.001) & \textbf{0.994 (0.004)} & \textbf{0.994 (0.004)} & 0.339 (0.009) & 0.011 (0.002)  \\
		& \texttt{se} & 0.818 (0.004) & 0.879 (0.003) & 0.791 (0.008) & \textbf{1.000 (0.000)} & 0.782 (0.002)& 0.779 (0.003) & 0.944 (0.007) & \textbf{0.998 (0.001)}\\ [-1.2em]\\ \hline \\[-1.2em] 
				\multirowcell{2}{\texttt{Model 3}}& \texttt{sp}& \textbf{1.000 (0.002)} & 0.991 (0.004) & 0.235 (0.014) & 0.010 (0.003) & 0.998 (0.001) & \textbf{1.000 (0.003)} & 0.658 (0.007) & 0.004 (0.001)  \\
		&\texttt{se} &\textbf{1.000 (0.001)} & \textbf{1.000 (0.000)} & \textbf{1.000 (0.000)} & \textbf{1.000 (0.000)} &  0.999 (0.001) & \textbf{1.000 (0.000)} & \textbf{1.000 (0.000)} &  \textbf{1.000 (0.000)}\\[-1.2em]\\ \hline\\ [-1.2em]
				\multirowcell{2}{\texttt{Model 4}}& \texttt{sp} &0.987 (0.007) & \textbf{0.998 (0.005)}& 0.248 (0.012) & 0.009 (0.003) & \textbf{0.994 (0.002)}& 0.993 (0.002) & 0.326 (0.008) & 0.011 (0.001)  \\
		&\texttt{se} & 0.693 (0.106) & 0.691 (0.103) & 0.995 (0.008) & \textbf{1.000 (0.001)} & 0.810 (0.003) & 0.807 (0.004) & \textbf{1.000 (0.001)} & \textbf{1.000 (0.000)}\\[-1.2em]\\ \hline \\[-1.2em]
				\multirowcell{2}{\texttt{Model 5}}& \texttt{sp}& \textbf{0.997 (0.006)} & 0.986 (0.007) & 0.589 (0.017) &  0.002 (0.001)&\textbf{0.981 (0.005)} & 0.980 (0.004)& 0.725 (0.009)&0.003 (0.001)  \\
		& \texttt{se}& \textbf{1.000 (0.001)} & \textbf{1.000 (0.001)} & \textbf{1.000 (0.000)} & \textbf{1.000 (0.000)} & 0.851 (0.002) & 0.849 (0.002)& \textbf{1.000 (0.000)} & \textbf{1.000 (0.000)}\\ [-1.2em]\\
		\hline\hline 
	\end{tabular}}
\end{sidewaystable}

\begin{sidewaystable}
\centering
  \caption{\texttt{\small rmse, mnorm,} and \texttt{\small 2norm} under the covariance structures in  \texttt{\small Model 1, 2, 3, 4}, and \texttt{\small 5}.}
 \label{tb:sim2}
 \scalebox{0.78}{
 \begin{tabular}{cccccc cccc}
		\hline\hline \\ [-1.7em] 
		\multirowcell{2}{}& \multirowcell{2}{Measure} & \multicolumn{4}{c}{p=50\, (n=100)} &\multicolumn{4}{c}{p=100\, (n=200)}\\ [-1.7em]\\\cmidrule(lr){3-6}\cmidrule(lr){7-10}  \\ [-1.7em]
		& & Proposed (MPM) & Proposed (MAP) & GL & Samp & Proposed (MPM) & Proposed (MAP) & GL &  Samp\\  [-1.7em]\\ \hline\hline
		\\[-1.2em]
		\multirowcell{3}{\texttt{Model 1}}& \texttt{rmse} & \textbf{0.081 (0.005)} & 0.094 (0.006) & 0.139 (0.007) & 0.256 (0.008) & 0.074 (0.004) & \textbf{0.072 (0.004)}& 0.084 (0.003)& 0.222 (0.004)\\ 
		& \texttt{mnorm}& 1.000 (0.104) & 1.000 (0.104)& 1.172 (0.094)& \textbf{0.979 (0.145)}& 1.607 (0.093) & 1.607 (0.093)& 1.169 (0.086)& \textbf{0.916 (0.001)}\\ 
		& \texttt{2norm}& \textbf{1.365 (0.187)} & 1.698 (0.192) & 2.434 (0.170)& 4.724 (0.456)& 1.917 (0.172)&\textbf{ 1.801 (0.175)}& 2.249 (0.147) & 6.129 (0.379) \\[-1.2em]\\ \hline\\
		[-1.2em]
		\multirowcell{3}{\texttt{Model 2}}& \texttt{rmse} & 0.141 (0.002)& 0.129 (0.003) & \textbf{0.106 (0.002)} & 0.253 (0.004)& 0.057 (0.002) & 0.057 (0.001)& \textbf{0.051 (0.001)}& 0.071 (0.001) \\ 
		& \texttt{mnorm}& 1.714 (0.042) & 1.714 (0.042) & \textbf{0.057 (0.019)} & 1.046 (0.092) & 0.319 (0.038) &0.323 (0.037) & 0.422 (0.025) &\textbf{ 0.297 (0.036)} \\ 
		& \texttt{2norm}& 3.825 (0.105)& \textbf{3.397 (0.117)} & 4.577 (0.069)& 7.555 (0.490)& 2.556 (0.134)& 2.520 (0.134) & \textbf{1.975 (0.078)} & 2.079 (0.159)\\[-1.2em]\\ \hline\\
		[-1.2em]
		\multirowcell{3}{\texttt{Model 3}}& \texttt{rmse} & \textbf{0.020 (0.002)} & 0.021 (0.001) & 0.054 (0.002) & 0.084 (0.003) & \textbf{0.012 (0.001)} & \textbf{0.012 (0.001)} & 0.034 (0.001)& 0.058 (0.001) \\ 
		& \texttt{mnorm}& 0.228 (0.038) & \textbf{0.213 (0.036)} &0.351 (0.034) & 0.314 (0.043) & 0.329 (0.025) & 0.329 (0.025)& 0.319 (0.024) & \textbf{0.237 (0.032)}\\ 
		& \texttt{2norm}& 0.380 (0.023) & \textbf{0.365 (0.023)}& 0.966 (0.055)& 1.677 (0.179)& \textbf{0.467 (0.017)}& \textbf{0.467 (0.017)} & 0.900 (0.035) & 1.735 (0.100)\\[-1.2em]\\ \hline\\
		[-1.2em]
		\multirowcell{3}{\texttt{Model 4}}& \texttt{rmse} & \textbf{0.057 (0.002)} & \textbf{0.057 (0.002)}& 0.073 (0.003)& 0.102 (0.005)& \textbf{0.036 (0.004)}& \textbf{0.036 (0.004)}& 0.050 (0.001)& 0.071 (0.001)\\ 
		& \texttt{mnorm}& 0.500 (0.001) & 0.500 (0.001) & 0.447 (0.037) &\textbf{0.391 (0.061)} & 0.500 (0.002) & 0.500 (0.001) & 0.424 (0.022)& \textbf{0.291 (0.031)}\\ 
		& \texttt{2norm}& 1.067 (0.057) & \textbf{1.066 (0.055)} & 1.353 (0.079)& 2.145 (0.237)& 1.088 (0.030)& \textbf{1.085 (0.032)} & 1.309 (0.048)& 2.098 (0.135)\\[-1.2em]\\ \hline\\
		[-1.2em]
		\multirowcell{3}{\texttt{Model 5}}& \texttt{rmse} & \textbf{0.136 (0.005)} & 0.139 (0.006)& 0.211 (0.009)& 0.355 (0.013)& 0.152 (0.004)& 0.142 (0.005)& \textbf{0.129 (0.004)}& 0.252 (0.005)\\ 
		& \texttt{mnorm}& \textbf{1.210 (0.100)}& 1.310 (0.105)& 1.699 (0.115)& 1.366 (0.207)& 1.114 (0.084)& 1.114 (0.085) & 1.498 (0.092)& \textbf{1.061 (0.129)}\\ 
		& \texttt{2norm}& \textbf{2.347 (0.095)}& 2.436 (0.102)& 3.340 (0.217)& 7.155 (0.818)&  3.544 (0.094)& 3.602 (0.103)& \textbf{2.828 (0.126)}& 7.463 (0.444)\\[-1.2em]\\
		\hline\hline 
	\end{tabular}}
\end{sidewaystable}

\subsection{Real dataset}

In this section, we assess the performance of the proposed estimator for linear discriminant analysis (LDA). We consider the breast cancer diagnostic dataset introduced by \cite{william1995}. The dataset is available from \texttt{\small http://archive.ics.uci.edu/ml}. The dataset consists of 30 numeric features extracted from a digitized image of a fine needle aspirate (FPA) of a breast mass on 212 malignant individuals and 357 benign ones. We randomly split the data 10 times into training set and test set, where training set consists of 72 malignant cases and 119 benign cases and test set consists of the remaining cases. Consider malignant case as class $1$ and benign case as class $0$. The LDA rule for observation $X_i$ ($i=1,2,\dots,378$) in test set is given by 
 \begin{align*}
     \delta(X_i)=\argmax_{j=0,1}\left\{X_i^t\hat{\bsSig}^{-1}\hat{\mu}_j-\frac{1}{2}\hat{\mu}_j^t\hat{\bsSig}^{-1}\hat{\mu}_j+\log \hat{\pi}_j\right\},
 \end{align*}
 where $\hat{\bsSig}$ is the estimated covariance matrix based on train set, $\hat{\mu}_j$ is the sample mean of class $j$ among train set, and $\hat{\pi}_j$ is the proportion of class $j$ among train set. 
 
 The Table \ref{tb:classification} shows the mean of classification error rate over 10 replications and the value in parentheses denotes the standard deviation. 
 
 \begin{table}[!ht]
	\centering
	\begin{tabular}{cccc}
		\hline \hline 
	    Proposed (MPM) & Proposed (MAP) & GL & Samp \\ \hline
	    \multirowcell{2}{$\bf 0.066$\\$\bf (0.016)$}& \multirowcell{2}{$\bf 0.066$\\$\bf (0.016)$} & \multirowcell{2}{$0.072$\\$(0.017)$}& \multirowcell{2}{$ 0.077$\\ $ (0.019)$}\\ \\
        \hline\hline
	\end{tabular}
		\caption{Classification error rate for breast cancer dataset}
		\label{tb:classification}
\end{table}

The result implies that the proposed estimator outperforms covariance graphical lasso and sample covariance matrix when it is applied to LDA classification. 

\section{Discussion}\label{sec:discuss}

In this paper, we propose a theoretically well supported method for estimating sparse covariances. We propose the method that uses Laplace approximation to calculate posterior model probabilities of covariance structures, or equivalently graphical structures, induced from spike and slab prior, generates MCMC samples for graphs using Metropolis-Hasting algorithm and approximated posterior model probabilities, and chooses the final model by MPM or MAP. We estimate the covariance matrix by the mode of conditional posterior of covariance matrix given the chosen model. We show that the error due to Laplace approximation becomes asymptotically marginal at some rate depending on posterior convergence rate under regular conditions on parameters. We propose a block coordinate descent algorithm to estimate covariance matrix when the covariance structure is given, and discuss its convergence. By simulation study based on five numerical models, we show that the proposed estimator performs better than graphical lasso and sample covariance matrix and is effective in introducing sparsity to the covariance structures. Also, the breast cancer dataset shows that the proposed estimator outperforms graphical lasso and sample covariance matrix when it is applied to LDA classification. 

\section{Acknowledgement}
Bongjung Sung was supported by an undergraduate research internship in the first half of 2021 Seoul National University College of Natural Sciences. Jaeyong Lee was supported by the National Research Foundation of Korea (NRF) grant funded by the Korea govern- ment(MSIT) (No. 2018R1A2A3074973 and 2020R1A4A1018207). 
\begin{appendices}
\section{Proof of Theorem \ref{thm:error_rate}}\label{appendixA}

\begin{lemma}\label{prodpd}
Assume $0<\beta<1/2$ in (A1) and (A2). Suppose that $\eta_n=o(1)$. Then, under prior (7),
\begin{align*}
    ||\bsSig_{\calZ}^*-\bsSig_0||_{\frob}=\bigO_p\left(\eta_n\right).
\end{align*}
\end{lemma}
 
\begin{proof}[\textbf{\upshape Proof}.]
Without loss of generality, we prove when $\calZ=\mathbf{1}_{p(p-1)/2}$, where $\mathbf{1}_{p(p-1)/2}$ is $p(p-1)/2-$vector with all entries being $1$, since the proof is essentially the same for general $\calZ$. The difference is only fixed zero entries of $\bsSig_{\calZ}$. In the proof, we write $\bsSig_{\calZ}$ as $\bsSig$ for simplicity suppressing dependency on $\calZ$. We prove Lemma \ref{prodpd} using the techniques considered in the proof of Theorem 1 in \cite{rothman2008}. 

For $\bfE=(e_{ij})\in\calM$, let $\bfE^{+}$ denote $p\times p$ diagonal matrix with diagonal entries being those in $\bfE$ and $\bfE^{-}=\bfE-\bfE^{+}$. Note that $\bfE^{-}$ is $p\times p$ matrix with diagonal entries being 0 and off-diagonal entries being those in $\bfE$. Also, let $||\bfE||_1=\sum_{i,j}|e_{ij}|$. For fixed $\calZ=\mathbf{1}_{p(p-1)/2}$ and $\bfX_n$, let 
\begin{align*}
    Q\left(\bsPsi\right)&=r_{\calZ}\left(\bsSig_0+\bsPsi,\bfX_n\right)-r_{\calZ}\left(\bsSig_0,\bfX_n\right)\\
    &=\left(\log\left|\bsSig_0+\bsPsi\right|-\log\left|\bsSig_0\right|\right)+\left(\trace\left(\bfS[\bsSig_0+\bsPsi]^{-1}\right)-\trace\left(\bfS\bsSig_0^{-1}\right)\right)\\
    &\quad+1/nv^2(||\left(\bsSig_0+\bsPsi\right)^{-}||_{\frob}^2-||\bsSig_0^{-}||_{\frob}^2)+\lambda/n\trace\left(\bsPsi^{+}\right), 
\end{align*}
where $r_{\calZ}(\cdot,\cdot)$ is defined in (10). Define the set 
\begin{align*}
    \bszeta=\{\calF\in\calM: ||\calF||_{\frob}\leq M\eta_n\},
\end{align*}
,where $M>0$ is constant to be determined. Put $\interior{\bszeta}=\text{int}\bszeta$ and $\partial \bszeta=\bszeta\setminus\interior{\bszeta}$. Observe that $Q$ is convex on the set $\bszeta$ with probability tending to one. To verify this, note that Hessian matrix of $Q$ is 
\begin{align*}
 \left(2[\bsSig_0+\bsPsi]^{-1}\bfS[\bsSig_0+\bsPsi]^{-1}-[\bsSig_0+\bsPsi]^{-1}\right)\otimes [\bsSig_0+\bsPsi]^{-1}+\Upsilon,
\end{align*}
where $\bsOmega_0=\bsSig_0^{-1}$ and $\Upsilon$ is $p(p+1)/2 \times p(p+1)/2$ diagonal matrix with diagonal entries being $0$ or $1/nv^2$. Observe that $2[\bsSig_0+\bsPsi]^{-1}\bfS[\bsSig_0+\bsPsi]^{-1}-[\bsSig_0+\bsPsi]^{-1}=[\bsSig_0+\bsPsi]^{-1}(2\bfS-(\bsSig_0+\bsPsi))[\bsSig_0+\bsPsi]^{-1}$. $\bfS$ converges to $\bsSig_0\in\calU(s_0,\tau_0)$ under the Frobenius norm in $\bbP_0-$probability, because $||\bfS-\bsSig_0||_{\frob}=O_p(p/\sqrt{n})$ and $p/\sqrt{n}$ tends to 0 as $p\asymp n^{\beta}$ for some constant $0<\beta<1/2$ by the assumption that $0<\beta<1/2$ in (A1). Also, since $\eta_n=o(1)$, if $\bsPsi\in \bszeta$, this implies that $2\bfS-(\bsSig_0+\bsPsi)$ converges to $\bsSig_0\in \calU(s_0,\tau_0)$. Furthermore, $\bsSig_0+\bsPsi$ converges to $\bsSig_0$ as $n\rightarrow\infty$, if $\bsPsi\in\bszeta$, which holds by $\eta_n=o(1)$. Thus, if $\bsPsi\in \bszeta$, 
$2[\bsSig_0+\bsPsi]^{-1}\bfS[\bsSig_0+\bsPsi]^{-1}-[\bsSig_0+\bsPsi]^{-1}\in\calM^+$ with probability tending to one. Therefore, we see $\left(2[\bsSig_0+\bsPsi]^{-1}\bfS[\bsSig_0+\bsPsi]^{-1}-[\bsSig_0+\bsPsi]^{-1}\right)\otimes [\bsSig_0+\bsPsi]^{-1}$ is positive definite with probability tending to one and so Hessian matrix of $Q$ on the set $\bszeta$ is positive definite with probability tending to one because $\Upsilon$ is semi-positive definite. 

Suppose that $\inf_{\bsPsi\in\partial\bszeta}Q(\bszeta)>0$. Then, since $Q$ is uniquely minimized at $\hat{\bsPsi}=\bsSig^*-\bsSig_0$ and so $Q(\hat{\bsPsi})\leq Q(\bfzero)$, we see that $\hat{\bsPsi}\in\bszeta$, which ends the proof. Thus, we show that $\inf_{\bsPsi\in\partial\bszeta}Q(\bszeta)>0$. By Taylor's expansion of $a(t)=\log|\bsSig_0+t\bsPsi|$ and $b(t)=\trace\left(\bfS[\bsSig_0+t\bsPsi]^{-1}\right)$ around 0, 

\begin{align}\label{A1}
\begin{split}
    \log|\bsSig_0+\bsPsi|-\log|\bsSig_0|&=\trace(\bsPsi\bsOmega_0)-\tilde{\bsPsi}^t\int_{0}^1(1-u)(\bsSig_0+u\bsPsi)^{-1}\otimes(\bsSig_0+u\bsPsi)^{-1}du \tilde{\bsPsi},\\
    \trace\left(\bfS(\bsSig_0+\bsPsi)^{-1}\right)-\trace\left(\bfS\bsSig_0^{-1}\right)&=-\trace\left(\bfS\bsOmega_0\bsPsi\bsOmega_0\right)\\
    &+\tilde{\bsPsi}^t\int_{0}^1(1-u)[2(\bsSig_0+u\bsPsi)^{-1}\bfS(\bsSig_0+u\bsPsi)^{-1}]\otimes(\bsSig_0+u\bsPsi)^{-1}du \tilde{\bsPsi}
\end{split}
\end{align}
where $\bsOmega_0=\bsSig_0^{-1}$ and $\tilde{\bsPsi}=\vectorize(\bsPsi)$. Also, 
\begin{align}\label{A2}
\begin{split}
    |\,||(\bsSig_0+\bsPsi)^{-}||_{\frob}^2-||\bsSig_0^{-}||_{\frob}^2|&=|\,||\bsSig_0^{-}+\bsPsi^{-}||_{\frob}^2-||\bsSig_0^{-}||_{\frob}^2|\\
    &=|\,||\bsSig_0^{-}||_{\frob}^2+2\sum_{i\neq j}\sgij^0\psiij+||\bsPsi^{-}||_{\frob}^2-||\bsSig_0^{-}||_{\frob}^2|\\
    &=|\,2\sum_{i\neq j}\sgij^0\psiij+||\bsPsi^{-}||_{\frob}^2|\\
    &\leq 2\tau\sum_{i\neq j}|\psiij|+||\bsPsi^{-}||_{\frob}^2\\
    &\leq 2\tau\sqrt{p^2-p}\sqrt{\sum_{i\neq j}\psiij^2}+||\bsPsi^{-}||_{\frob}^2\\
    &\leq 2\tau p||\bsPsi^{-}||_{\frob}+||\bsPsi^{-}||_{\frob}^2,
\end{split}
\end{align}
where $\bsPsi=(\psi_{ij})$, the first inequality holds by triangle inequality, $\bsSig_0\in \calU(s_0,\tau_0)$, and $\tau_0<\tau$, and the second inequality holds by Cauchy-Schwartz inequality and 
\begin{align}\label{A3}
\begin{split}
    |\trace(\bsPsi^{+})|&\leq ||\bsPsi^{+}||_1\leq \sqrt{p}||\bsPsi^{+}||_{\frob}\leq\sqrt{p+s_0}||\bsPsi^{+}||_{\frob}.
\end{split}
\end{align}
Note that the second inequality in (A.3) holds by Cauchy-Schwartz inequality. By (A.1)-(A.3), we have 
\begin{align}\label{A4}
\begin{split}
    Q(\bsPsi)&\geq \trace(\bsOmega_0[\bsSig_0-\bfS]\bsOmega_0\bsPsi)+P(\bsPsi)-\frac{1}{nv^2}(2\tau p||\bsPsi^{-}||_{\frob}+||\bsPsi^{-}||_{\frob}^2)-\frac{\lambda}{n}\sqrt{p+s_0}||\bsPsi^{+}||_{\frob},
\end{split}
\end{align}
where $P(\bsPsi)=\tilde{\bsPsi}^t\int_{0}^1(1-u)[2(\bsSig_0+u\bsPsi)^{-1}\bfS(\bsSig_0+u\bsPsi)^{-1}-(\bsSig_0+u\bsPsi)^{-1}]\otimes(\bsSig_0+u\bsPsi)^{-1}du \tilde{\bsPsi}$. Also,
\begin{align}\label{A5}
\begin{split}
    |\trace(\bsOmega_0[\bsSig_0-\bfS]\bsOmega_0\bsPsi)|&\leq ||\bsOmega_0[\bsSig_0-\bfS]\bsOmega_0||_{\frob}||\bsPsi||_{\frob}\\
    &\leq ||\bsOmega_0||_2^2||\bsSig_0-\bfS||_{\frob}||\bsPsi||_{\frob}\\
    &\leq D\tau^2\frac{p}{\sqrt{n}}||\bsPsi||_{\frob}\\
    &\leq D\tau^2\frac{p}{\sqrt{n}}(||\bsPsi^+||_{\frob}+||\bsPsi^-||_{\frob}),
\end{split}
\end{align}
for some constant $D>0$ and all sufficiently large $n$. The second inequality holds by (2) and the third inequality holds because $||\bfS-\bsSig_0||_{\frob}=\bigO_p(p/\sqrt{n})$. Finally, the last inequality holds by triangle inequality.Recall that $\eigmin(\bfA)=\inf_{||\bfx||_2=1}\bfx^t\bfA\bfx$ for $\bfA\in\calM^+$. Also, if $\bsPsi\in\bszeta$, 
\begin{align}\label{A6}
\begin{split}
 &\eigmin(\int_{0}^1(1-u)[2(\bsSig_0+u\bsPsi)^{-1}\bfS(\bsSig_0+u\bsPsi)^{-1}-(\bsSig_0+u\bsPsi)^{-1}]\otimes(\bsSig_0+u\bsPsi)^{-1}\diff u)\\
 &\geq\int_0^1(1-u)\eigmin(2(\bsSig_0+u\bsPsi)^{-1}\bfS(\bsSig_0+u\bsPsi)^{-1}-(\bsSig_0+u\bsPsi)^{-1})\eigmin((\bsSig_0+u\bsPsi)^{-1})\diff u  \\
 &\geq  \frac{1}{2}\inf_{\bsPsi\in\bszeta}\eigmin(2\bfS-(\bsSig_0+\bsPsi))\eigmin((\bsSig_0+\bsPsi)^{-1})^3  \\
 &\geq \frac{1}{2}\inf_{\bsPsi\in\bszeta}\eigmin(2\bfS-(\bsSig_0+\bsPsi))\inf_{\bsPsi\in\bszeta}\eigmin((\bsSig_0+\bsPsi)^{-1})^3
\end{split}
\end{align}
for all suffciently large $n$, where the first inequality holds by that fact that all eigenvalues of Kronecker product of $\bfA,\bfB\in\calM$ are in the form of product of an eigenvalue of $\bfA$ and that of $\bfB$ and the second inequality holds because $2\bfS-\bsSig_0\in\calM^+$ under the assumption that $p\asymp n^{\beta}$ for some constant $0<\beta<1/2$ as we have discussed.

Since $\eta_n=o(1)$ and $\bsPsi\in\bszeta$,
\begin{align}\label{A7}
\begin{split}
    \eigmin((\bsSig_0+\bsPsi)^{-1})^3&=[\eigmax(\bsSig_0+\bsPsi)]^{-3}\\
    &\geq \frac{1}{(||\bsSig_0||_2+||\bsPsi||_2)^3}\\
    &\geq \frac{1}{2}\tau^{-3}
\end{split}
\end{align}
for all sufficiently large $n$, where the first inequality holds by triangle inequality and
\begin{align}\label{A8}
    \eigmin(2\bfS-(\bsSig_0+\bsPsi))\geq \delta
\end{align}
for fixed sufficiently small constant $\delta>0$ with probability tending to one, as $2\bfS-(\bsSig_0+\bsPsi)$ converges to $\bsSig_0\in\calU(s_0,\tau_0)$ and $1/\tau_0\leq \eigmin(\bsSig_0)$ by the definition of $\calU(s_0,\tau_0)$. Thus. by (A.6)-(A.8), we have 
\begin{align}\label{A9}
\begin{split}
    P(\bsPsi)&\geq \frac{1}{4}\tau^{-3}\delta||\tilde{\bsPsi}||_2^2\\
    &=\frac{1}{4}\tau^{-3}\delta||\bsPsi||_{\frob}^2\\
    &=\frac{1}{4}\tau^{-3}\delta(||\bsPsi^+||_{\frob}^2+||\bsPsi^-||_{\frob}^2). 
\end{split}
\end{align}
By (A.4), (A.5) and (A.9), if $\bsPsi\in\partial\bszeta$, 
\begin{align}\label{A10}
\begin{split}
    Q(\bsPsi)&\geq \frac{1}{4}\tau^{-3}\delta(||\bsPsi^+||_{\frob}^2+||\bsPsi^-||_{\frob}^2)-D\tau^2\frac{p}{\sqrt{n}}(||\bsPsi^+||_{\frob}+||\bsPsi^-||_{\frob})\\
    &-\frac{1}{nv^2}(2\tau p||\bsPsi^{-}||_{\frob}+||\bsPsi^{-}||_{\frob}^2)-\frac{\lambda}{n}\sqrt{p+s_0}||\bsPsi^{+}||_{\frob} \\
    &=||\bsPsi^+||_{\frob}^2\left(\frac{1}{4}\tau^{-3}\delta-\left(D\tau^2\frac{p}{\sqrt{n}}+\frac{\lambda}{n}\sqrt{p+s_0}\right)\frac{1}{||\bsPsi^+||_{\frob}}\right)\\
    &+||\bsPsi^{-}||_{\frob}^2\left(\left(\frac{1}{4}\tau^{-3}\delta-\frac{1}{nv^2}\right)-\left(D\tau^2\frac{p}{\sqrt{n}}+\frac{2\tau p}{nv^2}\right)\frac{1}{||\bsPsi^{-}||_{\frob}}\right) \\
    &\geq ||\bsPsi^+||_{\frob}^2\left(\frac{1}{4}\tau^{-3}\delta-\frac{1}{M}\left(D\tau^2\sqrt{\frac{p}{(p+s_0)\log p}}+\lambda\sqrt{\frac{1}{np\log p}}\right)\right)\\
    &+||\bsPsi^{-}||_{\frob}^2\left(\left(\frac{1}{4}\tau^{-3}\delta-\frac{1}{nv^2}\right)-\frac{1}{M}\left(D\tau^2\sqrt{\frac{p}{(p+s_0)\log p}}+\frac{2\tau}{v^2}\sqrt{\frac{p}{n(p+s_0)\log p}}\right)\right).
\end{split}
\end{align}
Since $1<\tau$, $\tau=\bigO(1)$, $\lambda=\bigO(1)$, and $v$ is positive constant, if we choose constant $M>0$ to be sufficiently large, $\inf_{\bsPsi\in\partial\bszeta}Q(\bsPsi)>0$ by (A.10) for all sufficiently large $n$. Consequently, this gives Lemma \ref{prodpd}. 
\end{proof}

\begin{cor}\label{projection}
Assume that $0<\beta<1/2$ in (A1), (A2)-(A4) and  $\tau^4\leq p$, $\max\{1/\tau,1/p\}<\lambda<\log p/\tau_0$, $\tau>3$, $\tau^2\tau_0^2\leq s_0\log p$, and $n\geq s_0\log p/[(1-\tau_0/\tau)^2\tau^4]$. Suppose that $\eta_n=o(1)$. Then, under prior (7), 
\begin{align}\label{A11}
    ||\bsSig_{\calZ}^*-\bsSig_{\calZ}||_{\frob}=\bigO_p(\eta_n)
\end{align}
for all $\calZ$, where $\bsSig_{\calZ}\in\calU_{\calZ}(\tau)$.
\end{cor}
\begin{proof}
Note that the assumptions that $0<\beta<1/2$ in (A1), (A2)-(A4), $\tau^4\leq p$, $\max\{1/\tau,1/p\}<\lambda<\log p/\tau_0$, $\tau>3$, $\tau^2\tau_0^2\leq s_0\log p$, and $n\geq s_0\log p/[(1-\tau_0/\tau)^2\tau^4]$ imply Theorem \ref{thm:conv_rate} so that $||\bsSig_{\calZ}-\bsSig_0||_{\frob}=O_p(\epsilon_n)$. Since $\epsilon_n\leq \eta_n$ for all $n\geq 1$, this implies that $||\bsSig_{\calZ}-\bsSig_0||_{\frob}=\bigO_p(\eta_n)$. Because $\eta_n=o(1)$, we have $||\bsSig_{\calZ}^*-\bsSig_0||_{\frob}=\bigO_p(\eta_n)$ as a consequence of Lemma \ref{prodpd}. Therefore, by triangle inequality, 
\begin{align*}
    ||\bsSig_{\calZ}^*-\bsSig_{\calZ}||_{\frob}\leq ||\bsSig_{\calZ}-\bsSig_{0}||_{\frob}+||\bsSig_{\calZ}^*-\bsSig_{0}||_{\frob}=\bigO_p(\eta_n).
\end{align*}
\end{proof}

\begin{remark}
Note that (A.11) and (19) imply that 
\begin{align}\label{A12}
\frac{\int_{||\bsDelta_{\calZ}||_{\frob}\leq \eta_n} \exp\left(-\frac{n}{2}k_{\calZ}\left(\bsDelta_{\calZ},\bfX_n\right)\right)d\bsDelta_{\calZ}}{\int_{\bsSig_{\calZ}^*+\bsDelta_{\calZ} \in\calU_{\calZ}(\tau)}\exp\left(-\frac{n}{2}k_{\calZ}\left(\bsDelta_{\calZ},\bfX_n\right)\right)d\bsDelta_{\calZ}}\rightarrow 1
\end{align}
as $n\rightarrow\infty$, which plays a key role in the proof of Theorem \ref{thm:error_rate}. 
\end{remark}

\begin{lemma}\label{remainder}
Assume that $0<\beta<1/2$ in (A1), (A2), and $\eta_n=o(1)$. Then, under prior (7), 
\begin{align}\label{A21}
    |R_n|\leq (p+\#\calZ)\sum_{l=1}^4 \theta_{l}||\bsDelta_{\calZ}||_{\frob}^{l+2}
\end{align}
for some constants $\theta_l>0$ with probability tending to one, where $\bsDelta_{\calZ}=\bsSig_{\calZ}-\bsSig_{\calZ}^*$ and 
\begin{align*}
    R_n=k_{\calZ}(\bsDelta_{\calZ},\bfX_n)-k_{\calZ}(\bfzero_{p\times p},\bfX_n)-\frac{1}{2}\tilde{\bsDelta}_{\calZ}^t\hessian_{\bsSig_{\calZ}^*}\tilde{\bsDelta}_{\calZ}
\end{align*}
for $\tilde{\bsDelta}_{\calZ}=\vectorize(\bsDelta_{\calZ})$ and $\bfzero_{p\times p}$ is a $p\times p$ matrix with all entries being zero. 
\end{lemma}
\begin{proof}[\textbf{\upshape Proof}.]
Note that 
\begin{align*}
    \frac{\partial k_{\calZ}(\bsDelta_{\calZ},\bfX_n)}{\partial \tilde{\bsDelta}_{\calZ}}\bigg|_{\bsDelta_z=\bfzero_{p\times p}}=\bfzero_p,
\end{align*}
where $\bfzero_p$ is $p-$vector with all entries being $0$. Hence, $R_n$ can be seen as the remainder term in second order Taylor's expansion of $ k_{\calZ}(\bsDelta_{\calZ},\bfX_n)$ around $\bsDelta_{\calZ}=\bfzero_{p\times p}$, viewing $ k_{\calZ}(\bsDelta_{\calZ},\bfX_n)$ as a function of $\tilde{\bsDelta}_{\calZ}$ with $\bfX_n$ given. Then, following the argument in (A.14)-(A.16) of \cite{banerjee2015bayesian}, we have
\begin{align}\label{A13}
    |R_n|&\leq \frac{1}{2}||\bsDelta_{\calZ}||_{\frob}^2(p+\#\calZ)\max_{0\leq u\leq 1}||\hessian_{\bsSig^*+u\bsDelta_{\calZ}}-\hessian_{\bsSig_{\calZ}^*}||_{\infty}.
\end{align}
Write $u_{\calZ,ij}$ as $u_{ij}$ and $\omegaij^*$ as $\omega_{ij}$ for simplicity suppressing dependency on $\calZ$. Let $(\bsSig_{\calZ}^*+u\bsDelta_{\calZ})^{-1}=(a_{ij})$ and $(\bsSig_{\calZ}^*+u\bsDelta_{\calZ})^{-1}\bfS(\bsSig_{\calZ}^*+u\bsDelta_{\calZ})^{-1}=(b_{ij})$ suppressing dependency on $\calZ$. By (20), the entry of $\hessian_{\bsSig^*+u\bsDelta_{\calZ}}-\hessian_{\bsSig_{\calZ}^*}$ is in form of $(\sum a_{il}a_{jm}-\sum \omega_{il}^*\omega_{jm}^*)+(\sum b_{il}a_{jm} -\sum u_{il}\omega_{jm}^*)$. Consider $a_{il}a_{jm}-\omega_{il}^*\omega_{jm}^*$ and $a_{il}b_{jm}-u_{il}\omega_{jm}^*$ for example. Since 
\begin{align*}
    |a_{il}a_{jm}-\omega_{il}^*\omega_{jm}^*|&=|(a_{il}-\omega_{il}^*)(a_{jm}-\omega_{jm}^*)+(a_{il}-\omega_{il}^*)\omega_{jm}^*+(a_{jm}-\omega_{jm}^*)\omega_{il}^*|\\
    &\leq |a_{il}-\omega_{il}^*||a_{jm}-\omega_{jm}^*|+|a_{il}-\omega_{il}^*||\omega_{jm}^*|+|a_{jm}-\omega_{jm}^*||\omega_{il}^*|
\end{align*}
and
\begin{align*}
    |b_{il}a_{jm}-u_{il}\omega_{jm}^*|&=|(b_{il}-u_{il})(a_{jm}-\omega_{jm}^*)+(b_{il}-u_{il})\omega_{jm}^*+(a_{jm}-\omega_{jm}^*)u_{il}|\\
    &\leq |b_{il}-u_{il}||a_{jm}-\omega_{jm}^*|+|b_{il}-u_{il}||\omega_{jm}^*|+|a_{jm}-\omega_{jm}^*||u_{il}|,
\end{align*}
if we can bound $||(\bsSig_{\calZ}^*+u\bsDelta_{\calZ})^{-1}-\bsOmega_{\calZ}^*||_{\infty}$ and $||(\bsSig_{\calZ}^*+u\bsDelta_{\calZ})^{-1}\bfS(\bsSig_{\calZ}^*+u\bsDelta_{\calZ})^{-1}-\bfU_{\calZ}||_{\infty}$ by some polynomial with respect to $||\bsDelta_{\calZ}||_{\frob}$ with probability tending to one as in the proof of Lemma 4.3 in \cite{banerjee2015bayesian},  we can establish Lemma \ref{remainder}. Using the similar argument in (A.17)-(A.18) of \cite{banerjee2015bayesian}, it can be shown that
\begin{align}\label{A14}
    ||(\bsSig_{\calZ}^*+u\bsDelta_{\calZ})^{-1}-\bsSig_{\calZ}^*||_{\infty}&\leq K||\bsDelta_{\calZ}||_{\frob}
\end{align}
for some constants $K>0$ with probability tending to one. So it suffices to bound $||(\bsSig_{\calZ}^*+u\bsDelta_{\calZ})^{-1}\bfS(\bsSig_{\calZ}^*+u\bsDelta_{\calZ})^{-1}-\bfU_{\calZ}||_{\infty}$. By Woodbury's forumla, 
\begin{align*}
\begin{split}
    (\bsSig_{\calZ}^*+u\bsDelta_{\calZ})^{-1}\bfS(\bsSig_{\calZ}^*+u\bsDelta_{\calZ})^{-1}-\bfU_{\calZ}&=-u\bsOmega_{\calZ}^*\bsDelta_{\calZ}(\I+u\bsOmega_{\calZ}^*\bsDelta_{\calZ})^{-1}\bsOmega_{\calZ}^*\bfS\bsOmega_{\calZ}^*\\
    &\quad -u\bsOmega_{\calZ}^*\bfS \bsOmega_{\calZ}^*\bsDelta_{\calZ}(\I+u\bsOmega_{\calZ}^*\bsDelta_{\calZ})^{-1}\bsOmega_{\calZ}^*  \\
    &\quad +u^2\bsOmega_{\calZ}^*\bsDelta_{\calZ}(\I+u\bsOmega_{\calZ}^*\bsDelta_{\calZ})^{-1}\bsOmega_{\calZ}^*\bfS\bsOmega_{\calZ}^*\bsDelta_{\calZ}(\I+u\bsOmega_{\calZ}^*\bsDelta_{\calZ})^{-1}\bsOmega_{\calZ}^*.
\end{split}
\end{align*}
Hence, 
\begin{align}\label{A16}
\begin{split}
    ||(\bsSig_{\calZ}^*+u\bsDelta_{\calZ})^{-1}\bfS(\bsSig_{\calZ}^*+u\bsDelta_{\calZ})^{-1}-\bfU_{\calZ}||_{\infty}&\leq||(\bsSig_{\calZ}^*+u\bsDelta_{\calZ})^{-1}\bfS(\bsSig_{\calZ}^*+u\bsDelta_{\calZ})^{-1}-\bfU_{\calZ}||_2 \\
    &\leq u||\bsOmega_{\calZ}^*\bsDelta_{\calZ}(\I+u\bsOmega_{\calZ}^*\bsDelta_{\calZ})^{-1}\bsOmega_{\calZ}^*\bfS\bsOmega_{\calZ}^*||_2\\
    &+u||\bsOmega_{\calZ}^*\bfS\bsOmega_{\calZ}^*\bsDelta_{\calZ}(\I+u\bsOmega_{\calZ}^*\bsDelta_{\calZ})^{-1}\bsOmega_{\calZ}^*||_2\\
    &\quad +u^2||\bsOmega_{\calZ}^*\bsDelta_{\calZ}(\I+u\bsOmega_{\calZ}^*\bsDelta_{\calZ})^{-1}\bsOmega_{\calZ}^*\bfS\bsOmega_{\calZ}^*\bsDelta_{\calZ}\\
    &\quad (\I+u\bsOmega_{\calZ}^*\bsDelta_{\calZ})^{-1}\bsOmega_{\calZ}^*||_2\\
    &\leq 2u||\bsOmega_{\calZ}^*||_2^3||\bsDelta_{\calZ}||_2||\bfS||_2||(\I+u\bsOmega_{\calZ}^*\bsDelta_{\calZ})^{-1}||_2\\
    &+u^2||\bsOmega_{\calZ}^*||_2^4||\bsDelta_{\calZ}||_2^2||\bfS||_2||(\I+u\bsOmega_{\calZ}^*\bsDelta_{\calZ})^{-1}||_2\\
    &\leq (2\tau^4||\bsDelta_{\calZ}||_{\frob}+\tau^5||\bsDelta_{\calZ}||_{\frob}^2)||(\I+u\bsOmega_{\calZ}^*\bsDelta_{\calZ})^{-1}||_2,
\end{split}
\end{align}
where the first inequality and the last inequality holds by (2) and $u\leq 1$, the second inequality holds by triangle inequality, and the third inequality holds by the submultiplicativity of $||\cdot||_2$. Also, 
\begin{align}\label{A17}
\begin{split}
    ||(\I+u\bsOmega_{\calZ}^*\bsDelta_{\calZ})^{-1}||_2&=||((1-u)\bsSig_{\calZ}+u\bsSig_{\calZ}\bsOmega_{\calZ}^*\bsSig_{\calZ})^{-1}\bsSig_{\calZ}||_2\\
    &\leq ||((1-u)\bsSig_{\calZ}+u\bsSig_{\calZ}\bsOmega_{\calZ}^*\bsSig_{\calZ})^{-1}||_2||\bsSig_{\calZ}||_2\\
    &=\frac{||\bsSig_{\calZ}||_2}{||(1-u)\bsSig_{\calZ}+u\bsSig_{\calZ}\bsOmega_{\calZ}^*\bsSig_{\calZ}||_2}\\
    &\leq \frac{\tau}{||(1-u)\bsSig_{\calZ}+u\bsSig_{\calZ}\bsOmega_{\calZ}^*\bsSig_{\calZ}||_2},
\end{split}
\end{align}
where the first inequality holds by the submultiplicativity of $||\cdot||_2$ and the last equality holds because $(1-u)\bsSig_{\calZ}+u\bsSig_{\calZ}\bsOmega_{\calZ}^*\bsSig_{\calZ}\in\calM^+$ for all $0\leq u\leq 1$. Since
\begin{align}\label{A18}
\begin{split}
       ||(1-u)\bsSig_{\calZ}+u\bsSig_{\calZ}\bsOmega_{\calZ}^*\bsSig_{\calZ}||_2&\geq\eigmin((1-u)\bsSig_{\calZ}+u\bsSig_{\calZ}\bsOmega_{\calZ}^*\bsSig_{\calZ})\\
       &\geq (1-u)\eigmin(\bsSig_{\calZ})+u\eigmin(\bsSig_{\calZ}\bsOmega_{\calZ}^*\bsSig_{\calZ})\\
       &\geq (1-u)\eigmin(\bsSig_{\calZ})+u\eigmin^2(\bsSig_{\calZ})\eigmin(\bsOmega_{\calZ}^*)\\
       &\geq \frac{1-u}{\tau^3}+\frac{u}{\tau^3}=\frac{1}{\tau^3}.
\end{split}
\end{align}
By (A.17) and (A.18), it holds that $||(\I+u\bsOmega_{\calZ}^*\bsDelta_{\calZ})^{-1}||_2\leq \tau^4$ and so (A.16) implies  
\begin{align}\label{A19}
\begin{split}
   ||(\bsSig_{\calZ}^*+u\bsDelta_{\calZ})^{-1}\bfS(\bsSig_{\calZ}^*+u\bsDelta_{\calZ})^{-1}-\bfU_{\calZ}||_{\infty}&\leq 2\tau^8||\bsDelta_{\calZ}||_{\frob}+\tau^9||\bsDelta_{\calZ}||_{\frob}^2\\
   &\leq B(||\bsDelta_{\calZ}||_{\frob}+||\bsDelta_{\calZ}||_{\frob}^2),
\end{split}
\end{align}
for some constant $B>0$ with probability tending to one as $\tau=\bigO(1)$. Because $||\bfU_{\calZ}||_{\infty}\leq||\bfU_{\calZ}||_2\leq ||\bsOmega_{\calZ}^*||_2^2||\bfS||_2\leq\tau^3$ and $\tau=\bigO(1)$, by (A.15) and (A.19), 
\begin{align}\label{A20}
\begin{split}
    |(\sum a_{il}a_{jm}-\sum \omega_{il}^*\omega_{jm}^*)+(\sum b_{il}a_{jm} -\sum u_{il}\omega_{jm}^*)|&\leq\sum_{l=1}^4 \theta_l||\bsDelta_{\calZ}||_{\frob}^l. 
\end{split}
\end{align}
Thus, by (A.14) and (A.20), (A.13) holds with probability tending to one.
\end{proof}

\begin{lemma}\label{hessian_pd}
Assume $0<\beta<1/2$ in (A1) and (A2). Then the smallest eigenvalue of Hessian matrix $\hessian_{\bsSig_{\calZ}^*}$ of $k_{\calZ}$ in (19) is bounded away from 0 with probability tending to one.
\end{lemma}
\begin{proof}[\textbf{\upshape Proof}.]
Observe that $\hessian_{\bsSig_{\calZ}^*}$ can be expressed as $\bfV^t\left\{(2\bfU_{\calZ}-\bsOmega_{\calZ}^*)\otimes \bsOmega_{\calZ}^*\right\}\bfV+\bsPhi$, where $\bfV$ is $p^2\times (p+\#\calZ)$ matrix with full-rank whose entries are 0 or 1 and $\bsPhi$ is a $(p+\#\calZ)\times(p+\#\calZ)$ diagonal matrix with diagonal entries being $0$ or $1/nv^2$. Since $\bsPhi$ is semi-positive definite, it suffices to show that the smallest eigenvalue of $\bfV^t\left\{(2\bfU_{\calZ}-\bsOmega_{\calZ}^*)\otimes \bsOmega_{\calZ}^*\right\}\bfV$ is bounded away from $0$. Because $\bfV$ is slim matrix with full-rank, this can be established if the smallest eigenvalue of $(2\bfU_{\calZ}-\bsOmega_{\calZ}^*)\otimes \bsOmega_{\calZ}^*$ is bounded away from $0$. As
\begin{align*}
\begin{split}
    \eigmin((2\bfU_{\calZ}-\bsOmega_{\calZ}^*)\otimes \bsOmega_{\calZ}^*))&=\eigmin([\bsOmega_{\calZ}^*(2\bfS-\bsSig_{\calZ}^*)\bsOmega_{\calZ}^*]\otimes \bsOmega_{\calZ}^*)\\
    &=\eigmin(\bsOmega_{\calZ}^*(2\bfS-\bsSig_{\calZ}^*)\bsOmega_{\calZ}^*)\eigmin(\bsOmega_{\calZ}^*)\\
    &\geq \eigmin(\bsOmega_{\calZ}^*)^3\eigmin(2\bfS-\bsSig_{\calZ}^*)\\
    &\geq \frac{\delta}{\tau^3}>0
\end{split}
\end{align*}
for all sufficiently large $n$, Lemma \ref{hessian_pd} is established. Note that the last inequality holds because $\tau>1$ and $\tau=\bigO(1)$.
\end{proof}

\begin{proof}[\textbf{\upshape Proof of Theorem \ref{thm:error_rate}}.] Note that under the assumptions of Theorem \ref{thm:error_rate}, (A.12) and Lemma \ref{remainder} hold. Thus, as an analogy to the proof of Theorem 4.4 in \cite{banerjee2015bayesian}, it suffices to show that 
\begin{align}\label{A23}
    \frac{\int_{||\bsDelta_{\calZ}||_{\frob}\leq \eta_n}\exp\left(-\frac{n}{2}\left(\frac{1}{2}\tilde{\bsDelta}_{\calZ}^t\hessian_{\bsSig_{\calZ}^*}\tilde{\bsDelta}_{\calZ}+R_n\right)\right)d\bsDelta_{\calZ}}{\int_{||\bsDelta_{\calZ}||_{\frob}\leq \eta_n} \exp\left(-\frac{n}{4}\tilde{\bsDelta}_{\calZ}^t\hessian_{\bsSig_{\calZ}^*}\tilde{\bsDelta}_{\calZ}\right)d\bsDelta_{\calZ}}\rightarrow 1
\end{align}
as $n\rightarrow \infty$. Using the consequence of Lemma \ref{remainder} and following the argument in the proof of Theorem 4.4 in \cite{banerjee2015bayesian}, the ratio in (A.21) must lie between 
\begin{align*}
    [1\mp\eigmin(\hessian_{\bsSig_{\calZ}^*})^{-1}(p+\#\calZ)\eta_n]^{-(p+\#\calZ)/2}
\end{align*}
for all sufficiently large $n$ as $(p+\calZ)^2\eta_n=o(1)$. Also, $\eigmin(\hessian_{\bsSig_{\calZ}^*})$ is bounded away from $0$ with probability tending to one by Lemma \ref{hessian_pd}. Therefore, 
\begin{align*}
\begin{split}
     [1+\eigmin(\hessian_{\bsSig_{\calZ}^*})^{-1}(p+\#\calZ)\eta_n]^{-(p+\#\calZ)/2}
     &\geq [\exp\left(\eigmin(\hessian_{\bsSig_{\calZ}^*})^{-1}(p+\#\calZ)\eta_n\right)]^{-(p+\#\calZ)/2}\\
     &\geq \exp\left(-\eigmin(\hessian_{\bsSig_{\calZ}^*})^{-1}(p+\#\calZ)^2\eta_n/2\right)\rightarrow 1
\end{split}
\end{align*}
and
\begin{align*}
\begin{split}
     [1-\eigmin(\hessian_{\bsSig_{\calZ}^*})^{-1}(p+\#\calZ)\eta_n]^{-(p+\#\calZ)/2}&\leq [\exp\left(-2\eigmin(\hessian_{\bsSig_{\calZ}^*})^{-1}(p+\#\calZ)\eta_n\right)]^{-(p+\#\calZ)/2}\\
     &\leq \exp\left(\eigmin(\hessian_{\bsSig_{\calZ}^*})^{-1}(p+\#\calZ)^2\eta_n\right)\rightarrow 1
\end{split}
\end{align*}
as $n\rightarrow \infty$ because $(p+\#\calZ)^2\eta_n=o(1)$. Hence, we conclude that the error by Laplace approximation becomes negligible with probability tending to one under regular conditions.
\end{proof}

\section{Proof of Theorem \ref{laplace}}
In this section, we establish posterior convergence rate $\epsilon_n$ under prior (7) as in Theorem \ref{thm:conv_rate}. Define the set 
\begin{align*}
    B_{\epsilon_n}=\{p_{\bsSig}:K(p_{\bsSig_0},p_{\bsSig})\leq \epsilon_n^2,V(p_{\bsSig_0},p_{\bsSig})\leq\epsilon_n^2\},
\end{align*}
where $p_{\bsSig}$ is a probability density function of $\N_p(0,\bsSig)$ and 
\begin{align*}
    K(p_{\bsSig_0},p_{\bsSig})=\int p_{\bsSig_0}\log\frac{p_{\bsSig_0}}{p_{\bsSig}}, \quad V(p_{\bsSig_0},p_{\bsSig})=\int p_{\bsSig_0}\log^2\frac{p_{\bsSig_0}}{p_{\bsSig}}.
\end{align*}

 Let $\calP=\{p_{\bsSig}:\bsSig\in\calM^+\}$ be the space of all densities $p_{\bsSig}$ and consider a sieve $\calP_n=\{p_{\bsSig}:\bsSig\in\calM^+,s(\bsSig,\delta_n)\leq s_n,1/\tau\leq\eigmin(\bsSig)\leq\eigmax(\bsSig)\leq\tau,||\bsSig||_{\infty}\leq L_n\}\subset\calP$, where $\delta_n,s_n,$ and $L_n$ are to be determined in Theorem  \ref{packingnumber}. Denote the $\epsilon-$ packing number for subset $\calA$ of metric space $(\calS,d)$ by $\text{D}(\epsilon,\calA,d)$ , i.e., $\text{D}(\epsilon,\calA,d)$ is the minimum number of $d$-balls of size $\epsilon$ in $\calS$ needed to cover $\calA$ under metric $d$. Now, under regular conditions, we prove Theorem \ref{thm:conv_rate} by verifying conditions 
(10)-(12) of Lemma 5.1 in \cite{lee2021}, which is a version of Theorem 2.1 in \cite{ghosal2000convergence}. Define a function $s(\cdot,\cdot)$ on $\calM^+\times\bbR^+$ by $s(\calK,\delta)=\sum_{i<j}\dsone(|k_{ij}|\geq \delta)$, where $\calK=(k_{ij})$. Note that $s$ is equivalent to the number of edges in the graph induced by $\calK$ with a threshold $\delta$. Define the set 
\begin{align*}
    \calU(\delta_n,s_n,L_n,\tau)=\{\bsSig\in\calM^+:s(\bsSig,\delta_n)\leq s_n,1/\tau\leq\eigmin(\bsSig)\leq\eigmax(\bsSig)\leq\tau,||\bsSig||_{\infty}\leq L_n\}.
\end{align*}
Note that this set was also considered in \cite{lee2021}. We consider this set to obtain the upper bound of $\log D(\epsilon_n,\calP_n,d)$, since we are to use the argument in the proof of Theorem 5.2 in \cite{lee2021}. 

\begin{theorem}\label{packingnumber}
Assume (A1) and $\tau^4\leq p$. Let $s_n=c_1n\epsilon_n^2/\log p$, $L_n=c_2n\epsilon_n^2$, and $\delta_n=\epsilon_n/\tau^3$ for some constants $c_1>6$ and $c_2>0$. Suppose metric $d$ is Hellinger metric. Then, we have 
\begin{align*}
    \log \text{D}(\epsilon_n,\mathcal{P}_n,d)\lesssim n\epsilon_n^2.
\end{align*}
\end{theorem}

\begin{lemma}\label{lowerboundonsieve}
Assume (A1), (A2), (A4), and $\tau>3$, $1/\tau<\lambda<\log p/\tau_0$. Then 
\begin{align*}
   \pi^u(\bsSig\in \calU(\tau))\geq \exp(-2n\epsilon_n^2)
\end{align*}
for all sufficiently large $n$.
\end{lemma}

\begin{theorem}\label{sieve}
Assume (A1), (A2), (A4), and $\tau>3$, $\lambda<\log p/\tau_0$. Let $\delta_n$, $L_n$, and $s_n$ be those in Theorem \ref{packingnumber}. Then,
\begin{align*}
    \pi(\mathcal{P}_n^c)\leq \exp(-(c_1/3-2)n\epsilon_n^2)
\end{align*}
for all sufficiently large $n$. 
\end{theorem}

\begin{theorem}\label{neighborhood}
Assume (A1)-(A4) and $\tau^4\leq p$, $\tau^2\tau_0^2\leq s_0\log p$, $n\geq s_0\log p/[(1-\tau_0/\tau)^2\tau^4]$, $1/p<\lambda<\log p/\tau_0$. Then, for all sufficiently large $n$,
\begin{align*}
\pi(B_{\epsilon_n})\geq \exp\left(-\left(8+\frac{1}{\beta}\right)n\epsilon_n^2\right).
\end{align*}
\end{theorem}

\begin{proof}[\textbf{\upshape Proof of Theorem \ref{thm:conv_rate}}.]Under the assumptions (A1)-(A4) and some regular conditions on parameters, we see that Theorem \ref{packingnumber}, \ref{sieve}, and \ref{neighborhood} satisfy (10)-(12) of Lemma 5.1 in \cite{lee2021}, respectively. Thus, we see that 
\begin{align*}
\pi(d(p_{\bsSig_0},p_{\bsSig})\geq M\epsilon_n|\mathbf{X}_n)\rightarrow 0    
\end{align*}
for some constant $M>0$ under Hellinger metric $d$ as $n\rightarrow\infty$ in $\bbP_0-$probability. Note that the conditions on parameters can be held because we assume $\tau=\bigO(1)$, $\lambda=\bigO(1)$ and $v$ is some positive constant. Recall that we assumed $\epsilon_n=o(1)$. Thus, by Lemma A.1 (ii) in \cite{banerjee2015bayesian}, 
\begin{align*}
    ||\bsSig_0-\bsSig||_{\frob}&\leq ||\bsSig_0||_2||\bsSig||_2||\bsOmega_0-\bsOmega||_{\frob}\\
    &\leq \tau^2 ||\bsOmega_0-\bsOmega||_{\frob}\\
    &\leq B\tau^3 d(p_{\bsSig_0},p_{\bsSig})
\end{align*}
for some constant $B>0$ and all sufficiently large $n$, where $\bsOmega=\bsSig^{-1}$ and the first inequality holds by (2). So,
\begin{align}\label{B1}
 \frac{1}{B\tau^3}||\bsSig_0-\bsSig||_{\frob}\leq d(p_{\bsSig_0},p_{\bsSig}). 
\end{align}
for all sufficiently large $n$. Since $1<\tau$ and $\tau=O(1)$, (B.1) implies
\begin{align*}
    \pi(||\bsSig_0-\bsSig||_{\frob}\geq M'\epsilon_n|\bfX_n)\rightarrow 0,
\end{align*}
for some constant $M'>0$, which establishes Theorem \ref{thm:conv_rate}.
\end{proof}

\section{Proof of auxiliary results}

\begin{proof}[\textbf{\upshape Proof of Theorem \ref{packingnumber}}.] Following the argument in the proof of Theorem 5.2 in \cite{lee2021}, since $d(p_{\bsSig_1},p_{\bsSig_2})\leq C\tau^3 ||\bsSig_1-\bsSig_2||_{\frob}$ for some constant $C>0$,
\begin{align*}
    \log \text{D}(\epsilon_n,\mathcal{P}_n,d)&\leq \log \text{D}(\epsilon_n/(C\tau^3),\mathcal{U}(\delta_n,s_n,L_n,\tau),||\cdot||_{\text{F}})\\
    &\leq \log[\left(\frac{CL_n\tau^3}{\epsilon_n}\right)^p\sum_{j=1}^{s_n}\left(\frac{2CL_n\tau^3}{\epsilon_n}\right)^j\binom{\binom{p}{2}}{j}]\\
    &\leq p\log\left(\frac{CL_n\tau^3}{\epsilon_n}\right)+\log[\sum_{j=1}^{s_n}\left(\frac{2CL_n\tau^3}{\epsilon_n}\right)^{s_n}\binom{p+\binom{p}{2}}{s_n}]\\
    &=p\log\left(\frac{CL_n\tau^3}{\epsilon_n}\right)+\log s_n+s_n \log 2CL_n\tau^3+s_n\log \epsilon_n^{-1}+\log \binom{p+\binom{p}{2}}{s_n}\\
    &\lesssim p\log\left(\frac{CL_n\tau^3}{\epsilon_n}\right)+\log s_n+s_n \log 2CL_n\tau^3+s_n\log \epsilon_n^{-1}+s_n\log p,
\end{align*}
for all sufficiently large $n$, Note that the third inequality holds because $s_n\leq \binom{p}{2}/2$ for all sufficiently large $n$ so that $s_n=\bigO(p^{3/2})$. With simple calculations,
\begin{align}\label{C1}
    \frac{\log s_n+s_n \log 2CL_n\tau^3+s_n\log \epsilon_n^{-1}+s_n\log p}{n\epsilon_n^2}&=\frac{\log c_1(p+s_0)+c_1(p+s_0)\log \frac{2pCc_2(p+s_0)\log p \tau^3 }{n^{-\frac{1}{2}}(p+s_0)^{\frac{1}{2}}\log p^{\frac{1}{2}}}}{(p+s_0)\log p} \nonumber \\
    &=\frac{\log c_1(p+s_0)}{(p+s_0)\log p}+\frac{1}{2}\frac{c_1(p+s_0)\log n}{(p+s_0)\log p}\nonumber \\
    &+c_1\frac{\log 2Cc_2p}{\log p}+c_1\frac{\log(p+s_0)^{\frac{1}{2}}(\log p)^{\frac{1}{2}}\tau^3}{\log p}\nonumber\\
    &\asymp \frac{\log c_1(p+s_0)}{(p+s_0)\log p}+\frac{1}{2}\frac{c_1(p+s_0)\log p^{\frac{1}{\beta}}}{(p+s_0)\log p}+c_1\frac{\log 2Cc_2p}{\log p}\nonumber\\
    &+c_1\frac{\log(p+s_0)^{\frac{1}{2}}(\log p)^{\frac{1}{2}}}{\log p}\nonumber\\
    &\rightarrow\frac{c_1}{2\beta}+c_1+\frac{c_1}{2}=\left(\frac{3}{2}+\frac{1}{2\beta}\right)c_1
\end{align}
as $n\rightarrow \infty$. Here we used the assumption that $\tau_0^4\tau^2 \leq s_0\log p$, which implies $\tau^2 \leq s_0\log p$,  and $p\asymp n^{\beta}$ for the last inequality. Consequently, 
\begin{equation*}
    \log s_n+s_n \log 2CL_n\tau^3+s_n\log \epsilon_n^{-1}+s_n\log p\lesssim n\epsilon_n^2
\end{equation*}.

\noindent Also, by the assumption that $\tau^2\leq \tau_0^4\tau^2\leq s_0\log p$ and $p\asymp n^{\beta}$,
\begin{align*}
 \frac{p\log(\frac{CL_n\tau^3}{\epsilon})}{n\epsilon_n^2}&=\frac{p\log Cc_2\sqrt{n(p+s_0)\log p \tau^3}}{(p+s_0)\log p} \nonumber \\
 &\asymp \frac{\frac{p}{2\beta}\log p}{(p+s_0)\log p}+\frac{p\log Cc_2}{(p+s_0)\log p}+\frac{p\log \sqrt{(p+s_0)\log p}}{(p+s_0)\log p} \rightarrow \frac{1}{2\beta}.
\end{align*}
Thus we obtain that $p\log\left(\frac{CL_n\tau^3}{\epsilon_n}\right)\lesssim n\epsilon_n^2$. This, together with (C.1), gives 
\begin{align*}
    \log \text{D}(\epsilon_n,\mathcal{P}_n,d)\lesssim n\epsilon_n^2. 
\end{align*}
\end{proof}
\begin{proof}[\textbf{\upshape Proof of Lemma \ref{lowerboundonsieve}}.] Following the argument in the proof of Lemma 5.3 in \cite{lee2021} which uses Gershgorin circle theorem as in \cite{brualdi1994regions}, we see that 
\begin{align*}
    \pi^u(\bsSig\in\calU(\tau))\geq \pi^u(\tau^{-1}\leq \min_i(\sgii-\tau^{-1})\leq 2\max_i \sgii\leq \tau)\pi^u(\max_{i<j}|\sgij|<(\tau p)^{-1} ).
\end{align*}
Since we assume $\tau>3$ so that $\tau/4\geq 2\tau^{-1}$ as in the Lemma 5.3 of  \cite{lee2021},
\begin{align*}
        \pi^u(\tau^{-1}\leq \min_i(\sgii-\tau^{-1})\leq 2\max_i \sgii\leq \tau) &\geq \left[\frac{\lambda\tau}{8}\exp\left(-\frac{\lambda\tau}{4}\right)\right]^p\\
    &=\exp\left(-p\left(\frac{\lambda \tau}{4}-\log\left(\frac{\lambda\tau}{8}\right)\right)\right).
\end{align*}
Also, $1\leq\lambda\tau\leq\log p$ and so $\frac{\lambda \tau}{4}-\log\left(\frac{\lambda\tau}{8}\right)\leq \lambda\tau\leq \log p$. Thus,
\begin{align}\label{C2}
\begin{split}
  \pi^u(\tau^{-1}\leq \min_i(\sgii-\tau^{-1})\leq 2\max_i \sgii\leq \tau) &\geq\exp(-p\log p)\\
  &\geq \exp(-n\epsilon_n^2). 
\end{split}
\end{align}
for all sufficiently large $n$. Furthermore, one can see that 
\begin{align}\label{C3}
\begin{split}
    \pi^u(|\sgij|\geq(\tau p)^{-1})&=\pi^u(|\sgij|\geq(\tau p)^{-1}|\zij=0)\pi^u(\zij=0)+\pi^u(|\sgij|\geq(\tau p)^{-1}|\zij=1)\pi^u(\zij=1)\\
    &= 0\cdot(1-q)+\pi^u(|\sgij|\geq(\tau p)^{-1}|\zij=1)q\\
    &=\pi^u(|\sgij|\geq(\tau p)^{-1}|\zij=1)q\\
    &\leq q
\end{split}
\end{align}
for any $i\neq j$. Therefore, 
\begin{align}\label{C4}
\begin{split}
    \pi^u(\max_{i<j}|\sgij|<(\tau p)^{-1})&=\prod_{i<j}(1-\pi^u(|\sgij|\geq(\tau p)^{-1}))\\
    &\geq (1-q)^{p^2}\\
    &\geq \exp\left(-2\log p\right)\\
    &\geq \exp(-n\epsilon_n^2)
\end{split}
\end{align}
for all sufficiently large $n$. By (C.2) and (C.4), 
\begin{align*}
    \pi^u(\bsSig\in\calU(\tau))\geq \exp(-n\epsilon_n^2)\exp(-n\epsilon_n^2)=\exp(-2n\epsilon_n^2)
\end{align*}
for all sufficiently large $n$.
\end{proof}

\begin{proof}[\textbf{\upshape Proof of Theorem \ref{sieve}}.] We use the techniques in the proof of Theorem 5.4 in \cite{lee2021} to prove Theorem \ref{sieve}. Observe that 
\begin{align*}
    \pi(\calP_n^c)\leq \pi(s(\bsSig,\delta_n)>s_n)+\pi(||\bsSig||_{\infty}>L_n).
\end{align*}
Since $||\bsSig||_{\infty}\leq ||\bsSig||_2\tau=\bigO(1)$ and $L_n$ tends to $\infty$ as $n\rightarrow \infty$, we see that $\pi(||\bsSig||_{\infty}>L_n)=0$ for all sufficiently large $n$. So it suffices to bound $ \pi(s(\bsSig,\delta_n)>s_n)$. Following the argument in (C.3),
\begin{align*}
    \rho_n\equiv \pi^u(|\sgij|>\delta_n)&\leq q\asymp \frac{\log p}{p^2}
\end{align*}
for any $i<j$ and sufficiently large $n$. Hence, 
\begin{align}\label{C5}
     \binom{p}{2}\rho_n\lesssim \frac{p-1}{2p}\log p
     &< c_1(p+s_0)=s_n
\end{align}
for all sufficiently large $n$. Since (C.5) holds for all sufficiently large $n$, by Lemma A.3 in  \cite{song2018nearly} and the argument in the proof of Theorem 5.4 in \cite{lee2021}, 
\begin{align*}
    \pi(s(\bsSig,\delta_n)>s_n)\leq \frac{\exp(-\binom{p}{2}\text{H}(\rho_n||s_n/\binom{p}{2}))}{\sqrt{2\pi}\sqrt{2\binom{p}{2}\text{H}(\rho_n||s_n/\binom{p}{2})}},
\end{align*}
where
\begin{align*}
\text{H}(a||p)=a\log \frac{a}{p}+(1-a)\log \frac{1-a}{1-p}.
\end{align*}
Note that 
\begin{align}\label{C6}
\binom{p}{2}\text{H}(\rho_n||s_n/\binom{p}{2})=s_n\log(\frac{s_n}{\binom{p}{2}\rho_n})+\{\binom{p}{2}-s_n\}\log(\frac{\binom{p}{2}-s_n}{\binom{p}{2}-\binom{p}{2}\rho_n}).
\end{align}
Since $\rho_n\lesssim \log p/p^2$, we obtain lower estimate of the first term in (C.6) as following : 
\begin{align*}
    s_n\log (\frac{s_n}{\binom{p}{2}\rho_n})&\geq s_n \log(\frac{2c_1p(p+s_0)}{(p-1)\log p})\\
    &\geq s_n \log \sqrt{p} \\
    &=c_1/2n\epsilon_n^2
\end{align*}
for all sufficiently large $n$. Observe that 
\begin{align*}
    \frac{s_n-\binom{p}{2}\rho_n}{\binom{p}{2}(1-\rho_n)}\lesssim \frac{1}{1-\rho_n}\frac{\log p}{p}\rightarrow 0
 \end{align*} 
as $n\rightarrow \infty$ and note that $\log(1-x)\geq -2x$ for all sufficiently small $x>0$. Thus, using the similar argument as in the proof of Theorem 5.4 in \cite{lee2021}, we obtain lower estimate of the second term in (C.6) as following :
\begin{align*}
    (\binom{p}{2}-s_n)\log (\frac{\binom{p}{2}-s_n}{\binom{p}{2}-\binom{p}{2}\rho_n})\gtrsim -\frac{c_1n\epsilon_n^2}{\log p}
\end{align*}
for all sufficiently large $n$. Then clearly, 
\begin{align*}
    \binom{p}{2}\text{H}(\rho_n||s_n/\binom{p}{2})\geq c_1/2n\epsilon_n^2-\frac{c_1n\epsilon_n^2}{\log p}=c_1(\frac{1}{2}-\frac{1}{\log p})n\epsilon_n^2\rightarrow \infty 
\end{align*}
as $n\rightarrow \infty$. Consequently, 
\begin{align*}
    \pi^u(s(\bsSig,\delta_n)>s_n)&\leq \exp\left(-c_1\left(\frac{1}{2}-\frac{1}{\log p}\right)n\epsilon_n^2\right)\\
    &\leq \exp(-c_1n\epsilon_n^2/3)
\end{align*}
and by Lemma \ref{lowerboundonsieve}, 
\begin{align*}
    \pi(s(\bsSig,\delta_n)>s_n)&\leq \pi^u(s(\bsSig,\delta_n)>s_n)/\pi^u(\bsSig\in\calU(\tau))\\
    &\leq \exp(-c_1n\epsilon_n^2/3)\exp(2n\epsilon_n^2)\\
    &=\exp(-(c_1/3-2)n\epsilon_n^2)
\end{align*}
for all sufficiently large $n$. Note that $c_1/3-2>0$, since we chose $c_1$ to be larger than $6$.
\end{proof}
\begin{proof}[\textbf{\upshape Proof of Theorem \ref{neighborhood}}.] By Lemma 5.5 in  \cite{lee2021}, it suffices to show that 
\begin{align*}
    \pi(||\bsSig-\bsSig_0||_{\frob}\leq \sqrt{\frac{2}{3\tau^4\tau_0^2}}\epsilon_n)\geq \exp(-n\epsilon_n^2).
\end{align*}
Following the argument in the proof of Theorem 5.7 in \cite{lee2021}, this can be established if we show that 
\begin{align*}
 \pi(\calD_{\bsSig_0})\geq \exp\left(-\left(8+\frac{1}{\beta}\right)n\epsilon_n^2\right),    
\end{align*}
where
\begin{align*}
\calD_{\bsSig_0}=\{\max_{i<j}|\sgij-\sgij^0|\leq \sqrt{\frac{2s_0\log p}{{3np(p-1)\tau^4\tau_0^2}}},\max_i|\sgii-\sgii^0|\leq \sqrt{\frac{2\log p}{{3n\tau^4\tau_0^2}}}\}.    
\end{align*}
Since we assume $\bsSig_0\in\calU(\tau_0,s_0),\tau_0<\tau$ and $\tau^4(1-\tau_0/\tau)^2n\geq s_0\log p$ as in  \cite{lee2021}, one can see that if $\bsSig\in\calD_{\bsSig_0}$, $\bsSig\in\calU(\tau)$. Hence, as we are to obtain lower estimate of $\pi(\calD_{\bsSig_0})$, we assume independence on entries of $\bsSig$, as the constraint $\calU(\tau)$ on $\bsSig$ will only increase the prior concentration of $\calD_{\bsSig_0}$. Thus, it suffices to show that 
\begin{align*}
    \pi^u(\calD_{\bsSig_0})&\geq \exp\left(-\left(8+\frac{1}{\beta}\right)n\epsilon_n^2\right). 
\end{align*}
Observe that 
\begin{align}\label{C7}
    \pi^u(\calD_{\bsSig_0})&=\pi^u(\max_{i<j}|\sgij-\sgij^0|\leq \sqrt{\frac{2s_0\log p}{{3np(p-1)\tau^4\tau_0^2}}})\pi^u(\max_i|\sgii-\sgii^0|\leq \sqrt{\frac{2\log p}{{3n\tau^4\tau_0^2}}})\nonumber \\
    &= \underbrace{\prod_i \pi^u(\sgii^0-\sqrt{\frac{2\log p}{{3n\tau^4\tau_0^2}}}\leq \sgii\leq \sgii^0+\sqrt{\frac{2\log p}{{3n\tau^4\tau_0^2}}})}_{\text{\rom{1}}}\nonumber \\
&\underbrace{\prod_{i<j}\pi^u(\sgij^0- \sqrt{\frac{2s_0\log p}{{3np(p-1)\tau^4\tau_0^2}}}\leq \sgij\leq \sgij^0+ \sqrt{\frac{2s_0\log p}{{3np(p-1)\tau^4\tau_0^2}}})}_{\text{\rom{2}}}.
\end{align}
Since we assume $\log p/(\tau^4\tau_0^2)\leq n,\tau^4\leq p$ and $p^{-1}<\lambda<\log p/\tau_0$ as in  \cite{lee2021}, one can see that 
\begin{align}\label{C8}
\begin{split}
    \text{\rom{1}}&\geq \exp\left(-\left(3+\frac{1}{2\beta}\right)p\log p\right)\\
    &\geq \exp\left(-\left(3+\frac{1}{2\beta}\right)n\epsilon_n^2\right)
\end{split}
\end{align}
for all sufficiently large $n$. Also,
\begin{align*}
    \text{\rom{2}}&=\prod_{\sgij^0=0}\pi^u(-\sqrt{ \frac{2s_0\log p}{{3np(p-1)\tau^4\tau_0^2}}}\leq \sgij \leq \sqrt{\frac{2s_0\log p}{{3np(p-1)\tau^4\tau_0^2}}})\\
    &\prod_{\sgij^0\neq 0}\pi^u(\sgij^0- \sqrt{\frac{2s_0\log p}{{3np(p-1)\tau^4\tau_0^2}}}\leq \sgij\leq \sgij^0+ \sqrt{\frac{2s_0\log p}{{3np(p-1)\tau^4\tau_0^2}}})\\
    &=\underbrace{\prod_{\sgij^0=0}\pi^u(|\sgij|\leq \sqrt{\frac{2s_0\log p}{{3np(p-1)\tau^4\tau_0^2}}})}_{\text{\rom{3}}}\underbrace{\prod_{\sgij^0\neq 0}\pi^u(|\sgij-\sgij^0|\leq \sqrt{\frac{2s_0\log p}{{3np(p-1)\tau^4\tau_0^2}}})}_{\text{\rom{4}}}.
\end{align*}
Following the argument in (C.3),
\begin{align*}
    \pi^u(|\sgij|>\sqrt{\frac{2s_0\log p}{{3np(p-1)\tau^4\tau_0^2}}})\leq q \asymp \frac{\log p}{p^2}.
\end{align*}
Therefore, 
\begin{align}\label{C9}
\begin{split}
    \text{\rom{3}}&=\prod_{\sgij^0=0}\left(1-\pi^u(|\sgij|> \sqrt{\frac{2s_0\log p}{{3np(p-1)\tau^4\tau_0^2}}})\right)\\
    &\gtrsim \left(1-\frac{\log p}{p^2}\right)^{p^2}\\
    &\geq  \exp(-n\epsilon_n^2).
\end{split}
\end{align}
Now we obtain lower estimate of \rom{4}. Following the argument in (C.3), observe that 
\begin{align*}
\pi^u(|\sgij-\sgij^0|\leq \sqrt{\frac{2s_0\log p}{{3np(p-1)\tau^4\tau_0^2}}})&=\pi^u(|\sgij-\sgij^0|\leq \sqrt{\frac{2s_0\log p}{{3np(p-1)\tau^4\tau_0^2}}}|\zij=1)q
\end{align*}
for any $i<j$ such that $\sgij^0\neq 0$ and sufficiently large $n$, where 
the second inequality holds because $0<|\sgij^0|\leq \tau=\bigO(1)$ for all sufficiently large $n$. Thus for all sufficiently large $n$, 

\begin{align}\label{C10}
\begin{split}
    \text{\rom{4}}&=\prod_{\sgij^0\neq 0}q\int_{\sgij^0-\sqrt{\frac{2s_0\log p}{{3np(p-1)\tau^4\tau_0^2}}}}^{\sgij^0+\sqrt{\frac{2s_0\log p}{{3np(p-1)\tau^4\tau_0^2}}}}\N(\sgij|0,v^2)d\sgij\\
    &\geq \prod_{\sgij^0\neq0} \exp\left(-\log\left[\frac{\sqrt{2\pi}v}{2q}\sqrt{\frac{3np(p-1)\tau^4\tau_0^2}{2s_0\log p}}\right]-\frac{2\tau_0^2}{v^2}\right)\\ 
    &=\exp\left(-s_0\log\left[\frac{\sqrt{2\pi}v}{2q}\sqrt{\frac{3np(p-1)\tau^4\tau_0^2}{2s_0\log p}}\right]-s_0\frac{2\tau_0^2}{v^2}\right)\\
    &\geq \exp\left(-s_0 \log \sqrt{np^6}-s_0\frac{2\tau_0^2}{v^2}\right)  \\ 
    &= \exp\left(-\left(3+\frac{1}{2\beta}\right)s_0\log p-s_0\frac{2\tau_0^2}{v^2}\right) \\
    &\geq \exp\left(-\left(3+\frac{1}{2\beta}\right)(p+s_0)\log p-(p+s_0)\log p\right)\\
    &=\exp\left(-\left(4+\frac{1}{2\beta}\right)n\epsilon_n^2\right)
\end{split}
\end{align}
for all sufficiently large $n$, because $q\asymp \log p/p^2$, $|\sgij|\leq \tau$, and $2s_0\log p/(3np(p-1)\tau^4\tau_0^2)$ tends to 0 as $n\rightarrow \infty$. 
By (C.9) and (C.10),
\begin{align}\label{C11}
\text{\rom{2}}&=\text{\rom{3}}\cdot\text{\rom{4}}\geq\exp(-n\epsilon_n^2)\exp\left(-\left(4+\frac{1}{2\beta}\right)n\epsilon_n^2\right)=\exp\left(-\left(5+\frac{1}{2\beta}\right)n\epsilon_n^2\right)
\end{align}
for all sufficiently large $n$. Consequently, by (C.7), (C.8), and (C.11),
\begin{align*}
    \pi^u(\calD_{\bsSig_0})=\text{\rom{1}}\cdot \text{\rom{2}}\geq \exp\left(-\left(3+\frac{1}{2\beta}\right)n\epsilon_n^2\right)\exp\left(-\left(5+\frac{1}{2\beta}\right)n\epsilon_n^2\right)=\exp\left(-\left(8+\frac{1}{\beta}\right)n\epsilon_n^2\right),
\end{align*}
which concludes the proof. 
\end{proof}

\section{Proofs of propositions from Section \ref{bcd}}
\begin{proof}[\textbf{\upshape Proof of Proposition \ref{positivity}}.]
Define a function $m(\cdot)$ on $\bbR^p$ by
\begin{align*}
m(\bfx)=\bfx^t\bsSig_{11}^{-1}\bfS_{11}\bsSig_{11}^{-1}\bfx-2\bfx^t\bsSig_{11}^{-1}\bfs_{12}+s_{22}.    
\end{align*}
Note that $\bsSig_{11}^{-1}\bfS_{11}\bsSig_{11}^{-1}$ is positive definite with probability tending to one and so $m(\cdot)$ is convex. Thus, solving the equation $\frac{d m}{d \bfx}=\bfzero$, one can see that the global minimizer of $m(\cdot)$ is $\bfx=\bsSig_{11}\bfS_{11}^{-1}\bfs_{12}$ and $m(\hat{\bfx})=s_{22}-\bfs_{12}^t\bfS_{11}^{-1}\bfs_{12}$, which is Schur complement of $\bfS_{11}$. Since $\bfS$ is positive definite if and only if $\bfS_{11}$ and $s_{22}-\bfs_{12}^t\bfS_{11}^{-1}\bfs_{12}$ are positive definite by Corollary 14.2.14 of  \cite{harville2008}, $m(\hat{\bfx})$ is positive with probability tending to one. Therefore, with probability tending to one, $0<m(\hat{\bfx})\leq m(\bsbeta)=\bfu$.
\end{proof}
\begin{proof}[\textbf{\upshape Proof of Proposition \ref{convergence}}.]
Clearly, the initial covariance matrix is positive definite. Denote the matrix resulted from updating a single row/column of a previous matrix $\bsSig$ by $\hat{\bsSig}$. Partition $\bsSig$ and $\hat{\bsSig}$ as in (23). Provided that a previous matrix is positive definite, $\hat{\bsSig}_{11}$ and its Shur complement $\hat{\bsSig}_{11\cdot2}=\hat{\bsSig}_{22}-\hat{\bsSig}_{12}^\top \hat{\bsSig}_{11}^{-1}\hat{\bsSig}_{12}$ are positive definite by Proposition \ref{positivity} and the fact that $\hat{\bsSig}_{11}$ is the same as $\bsSig_{11}$, which is positive definite. Consequently, by Corollary 14.2.14 of  \cite{harville2008}, $\hat{\bsSig}$ is positive definite. Thus, the estimated covariance matrix by Algorithm \ref{algorithm2} is positive definite by an induction. 
 
 The convergence of Algorithm \ref{algorithm2} to the stationary point of the objective function in (17) can be established using the arguments considered by \cite{wang14} and Breheny and Huang \cite{breheny2011}. Both  Wang \cite{wang14} and Breheny and Huang \cite{breheny2011} used the results of Lemma 3.1 and part (c) of Theorem 4.1 in  \cite{tseng2001}. Although  \cite{tseng2001} considered non-convex and non-differentiable objective function, the sufficient conditions for Lemma 3.1 and part (c) of Theorem 4.1 in  \cite{tseng2001} can be extended to convex and continuously differentiable function.  
\end{proof}

\end{appendices}
\newpage

\bibliographystyle{elsarticle-harv}
\bibliography{CovLap}

\end{document}